\RequirePackage{ifthen}
\newif\ifreport
\reporttrue

\ifreport

\documentclass[a4paper]{llncs}
\usepackage[UKenglish]{babel}
\usepackage{amsmath,amssymb,stmaryrd,xspace,tikz,wasysym,url,hyperref}
\usepackage{algorithm, algorithmicx}
\usepackage{algpseudocode}
\usepackage{pgfplots}
\usepackage{pgfplotstable}
\usepackage{thm-restate}

\usetikzlibrary{arrows,automata,shapes,calc,fit}

\else

\documentclass{eptcs}
\usepackage{microtype}
\usepackage{amsmath,amssymb,stmaryrd,xspace,tikz,wasysym}
\usepackage{algorithm, algorithmicx}
\usepackage{algpseudocode}
\usepackage{pgfplots}
\usepackage{pgfplotstable}
\usepackage{thm-restate}

\usetikzlibrary{arrows,automata,shapes,calc,fit}

\fi

\newcommand{\myomit}[1]{}

\newcommand{\mycomment}[1]{}

\ifreport

\else

\newtheorem{defi}{Definition}
\newtheorem{theo}{Theorem}
\newtheorem{prop}{Proposition}
\newtheorem{lemm}{Lemma}
\newtheorem{coro}{Corollary}
\newtheorem{exam}{Example}

\newenvironment{definition}{\begin{defi} \rm }{\end{defi}}
\newenvironment{theorem}{\begin{theo} \rm }{\end{theo}}
\newenvironment{proposition}{\begin{prop} \rm }{\end{prop}}
\newenvironment{lemma}{\begin{lemm} \rm }{\end{lemm}}
\newenvironment{corollary}{\begin{coro} \rm }{\end{coro}}
\newenvironment{example}{\begin{exam} \rm }{\end{exam}}

\newenvironment{proof}{\begin{trivlist} \item[\hspace{\labelsep}\bf Proof:]}{\hfill$\Box$\end{trivlist}}
\fi

\newcommand{\tuple}[1]{\langle #1 \rangle}
\newcommand{\nontop}[1]{\textsc{NonTop}(#1)}

\def\nat{\mathbb{N}}
\def\eqdef{\stackrel{\triangle}{=}}
\def\vis{\textit{vis}}
\def\vislen{vislen}
\def\oddresponse{\textsc{OddResponse}}
\newcommand{\plays}[2]{\Pi(#1,#2)}

\renewcommand{\i}{\ensuremath{\ocircle}\xspace}
\newcommand{\odd}{\ensuremath{\square}\xspace}
\newcommand{\even}{\ensuremath{\Diamond}\xspace}
\newcommand{\player}{\ensuremath{\ocircle}\xspace}

\newcommand{\comment}[1]{\textcolor{blue}{[#1]}}

\newcommand{\solitaire}[1][]{\ensuremath{\mc{G}^{#1}}\xspace}

\newcommand{\post}[1]{\ensuremath{#1^{\bullet}}}

\usepackage{todonotes}

\newcommand{\pnot}[1]{\bar{#1}}

\newcommand{\domain}[1]{\textbf{dom}(#1)}

\newcommand{\attrsym}{\ensuremath{\textit{Attr}}}
\newcommand{\attr}[3][]{\ensuremath{{#2}{\text{-}}\attrsym^{#1}(#3)}}
\newcommand{\myattr}[3]{\ensuremath{{#2}{\text{-}}\attrsym^{#1}(#3)}}

\newcommand{\attrW}[3]{\ensuremath{\mathop{{#1}{\text{-}}\attrsym_{#3}(#2)}}}
\newcommand{\myattrW}[4]{\ensuremath{\mathop{{#2}{\text{-}}\attrsym^{#1}_{#4}(#3)}}}

\newcommand{\winsubodd}[1]{\textsf{Win}_{\odd}(#1)}
\newcommand{\winsubeven}[1]{\textsf{Win}_{\even}(#1)}
\newcommand{\minprio}[1]{\textsf{min}_{\priosym}(#1)}

\newcommand{\mc}[1]{\ensuremath{\mathcal{#1}}}
\newcommand{\ie}{\emph{i.e.}\xspace}
\newcommand{\eg}{\emph{e.g.}\xspace}
\newcommand{\viz}{\emph{viz.}\xspace}

\newcommand{\oftype}{{:}}

\newcommand{\priosym}{\mathcal{P}}
\newcommand{\prio}[1]{\priosym(#1)}

\newcommand{\strategy}[1]{\mathbb{S}_{#1}}
\newcommand{\memstrategy}[1]{\mathbb{S}_{#1}^*}

\ifreport

\begin{document}
\title{Strategy Derivation for Small Progress Measures}
\author{M. Gazda and T.A.C. Willemse
\institute{Eindhoven University of Technology, The Netherlands}
}

\else

\title{Improvement in Small Progress Measures}
\def\titlerunning{Improvement in Small Progress Measures}
\author{Maciej Gazda and Tim A.C. Willemse
\institute{Eindhoven University of Technology, The Netherlands}
}
\def\authorrunning{M. Gazda and T.A.C. Willemse}




\fi

\def\runtimeceil{O(dm \cdot (n/\lceil d / 2 \rceil)^{\lceil d/2 \rceil})}
\def\runtimefloor{O(dm \cdot (n/\lfloor d / 2 \rfloor)^{\lfloor d/2 \rfloor})}
\def\floord2{{\lfloor d/2 \rfloor}}
\def\ceild2{{\lceil d/2 \rceil}}

\ifreport
\maketitle 
\else

\begin{document}
\maketitle 

\fi

\begin{abstract}
Small Progress Measures is one of the most efficient parity game solving algorithms. The original algorithm provides the full solution (winning regions and strategies) in $\runtimeceil$ time, and requires a re-run of the algorithm on one of the winning regions. We provide a novel operational interpretation of progress measures, and modify the algorithm so that it derives the winning strategies for both players in one pass. This  reduces the upper bound on strategy derivation for SPM to $\runtimefloor$. 
\end{abstract}

\section{Introduction}

A parity game \cite{EJ:91,McN:93,Zie:98} is an infinite duration
game played on a directed graph by two players called \emph{even}
and \emph{odd}. Each vertex in the graph is owned by one of the
players, and labelled with a natural number, called a priority.
The game is played by pushing a token along the edges in the graph;
the choice where to move next is made by the owner of the vertex
on which the token currently resides.  The winner of the thus
constructed play is determined by the parity of the minimal (or
maximal, depending on the convention) priority that occurs infinitely
often, and the winner of a vertex is the player who has a \emph{strategy}
to force every play originating from that vertex to be winning for
her.  Parity games are determined; that is, each vertex is won by
some player ~\cite{McN:93}. \emph{Solving} a game essentially means
deciding which player wins which vertices in the game.

Parity games play an important role in several foundational results;
for instance, they allow for an elegant simplification of the hard
part of Rabin's proof of the decidability of a monadic second-order
theory, and a number of decision problems of importance can be
reduced to deciding the winner in parity games. For instance, the
model checking problem for the modal $\mu$-calculus is equivalent,
via a polynomial-time reduction, to the problem of solving parity
games~\cite{EJS:93,SS:98}; this is of importance in computer-aided
verification.  Winning strategies for the players play a crucial
role in supervisory control of discrete event systems, in which
such strategies are instrumental in constructing a supervisor that
controls a plant such that it reaches its control objectives and
avoids bad situations; see \eg~\cite{AVW:03} and the references
therein. In model checking, winning strategies are essential in
reporting witnesses and counterexamples, see~\cite{SS:98}.

A major impetus driving research in parity
games is their computational status. Even though the solution problem
belongs to both the complexity classes NP and coNP, no polynomial
algorithm has been devised so far. 
Taking a slightly simplified view, today's deterministic algorithms
for parity game solving can be classified into three categories: a
category of two early classical algorithms, \viz the \emph{recursive
algorithm}~\cite{Zie:98}, solving games with $d$ different priorities,
$n$ vertices and $m$ edges in $O(m\cdot n^d)$ and the \emph{small
progress measures} (SPM) algorithm~\cite{Jur:00}, solving games in
$\runtimefloor$; a category of the fastest known algorithms, \viz
the deterministic subexponential algorithm~\cite{JPZ:06} and the
\emph{bigstep} algorithm~\cite{Sch:07}; and a category of strategy improvement
algorithms \cite{VJ:00,Sch:08,Fea:10}.

For a considerable time, strategy improvement algorithms were
perceived as likely candidates for solving parity games in polynomial
time, but they were ultimately proven to be exponential in the
worst-case~\cite{Fri:11}.  In fact none of today's deterministic
strategy improvement algorithms matches the bigstep algorithm or
the deterministic subexponential algorithm.  The latter is a
modification of the classical recursive algorithm, running in
$n^{O(\sqrt{n})}$, and the bigstep algorithm combines the recursive
algorithm and the SPM algorithm, running in $O(m \cdot (\kappa n
/d)^{\gamma(d)})$, where $\kappa$ is a small constant and $\gamma(d)
\approx d/3$.

Somewhat surprisingly, our knowledge of the classical algorithms
is still far from complete. For instance, the recursive algorithm
is regarded as one of the best algorithms in practice, which is
corroborated by experiments~\cite{FL:09}. However, until our recent
work~\cite{GW:13} where we showed the algorithm is well-behaved on
several important classes of parity games, there was no satisfactory
explanation why this would be the case. In a similar vein, in
\emph{ibid.} we provided tighter bounds on the worst-case running time, but so far, no tight bounds
for this seemingly simple algorithm have been established.  We
expect that, if improvements on the upper bound on the parity game
solving problem can be made, such improvements will come from
improvements in, or through a better understanding of the classical
algorithms; this expectation is fuelled by the fact that these
classical algorithms are at the basis of the currently optimal
algorithms.

Here, we focus on the second classical algorithm, namely Jurdzi\'nski's
small progress measures algorithm. Using a fixpoint computation,
it computes a \emph{progress measure}, a labelling of vertices,
that witnesses the existence of winning strategies.   In general, no clear,
intuitive interpretation of the information contained in the progress
measures has been given, and the mechanics of the algorithm are
still quite mysterious.  This is in contrast to the self-explanatory
recursive algorithm, and the strategy improvement algorithm, where,
thanks to ordering of plays according to profiles, at every step,
one has a clear picture of the currently known best strategy. Apart
from Jurdzi{\'n}ski's original article, some additional insight was
offered in \cite{2001automata} (an intuitive progress measure in
the setting of solitaire games), and also in Schewe's paper on
\emph{bigstep} \cite{Sch:07} (restricted codomain and small dominions).
Our \emph{first} contribution  is to provide a better understanding of
these progress measures and the intermediate values in the fixpoint
computation, see Section~\ref{sec:interpretation}.  By doing so,
a better understanding of the algorithm itself is obtained.

Progress measures come in two flavours, \viz even-and odd-biased,
and their computation time is bounded by either $\runtimefloor$ or
$\runtimeceil$, depending on the parity of the extreme priorities.
From an even-biased progress measure, one immediately obtains winning
regions for \emph{both} players, but only a winning strategy for
player even on its winning region and not for her opponent. Likewise
for odd-biased progress measures. Obtaining the winning strategy
for an opponent thus requires re-running the algorithm on the
opponent's winning region.  Note that the effort that needs to be
taken to obtain a strategy in the same time bound as the winning
region stems from a more general feature of parity games: a winning
partition in itself does not allow one to efficiently compute a
winning strategy (unless there is an efficient algorithm for solving
parity games). This basic result, which we nevertheless were unable
to find in the literature, is formalised in Section \ref{sec:strategy}.

An essential consequence of this is that the algorithm solves a
parity game in $\runtimefloor$, as one can always compute one of
the two types of progress measures in the shorter time bound. Contrary
to what is stated in~\cite{Jur:00},
the same reasoning does not apply to computing the winning strategy for
a fixed player; this still requires $\runtimeceil$ in the worst
case,  as also observed by Schewe in~\cite{Sch:07}.
An intriguing open problem is
whether it is at all possible to derive the winning strategies for
both players while computing one type of measure only, as this would
lower the exponent in the time bound to $\floord2$.  Our \emph{second}
contribution is to
give an affirmative answer to the above problem.  We modify the
generic SPM by picking a partial strategy when a vertex won by
player $\odd$ is discovered, and subsequently adjust the lifting
policy so that it prioritises the area which contains an $\odd$-dominion.
Both steps are inspired by the interpretation of the progress measures
that we discuss in Section~\ref{sec:interpretation}.
Like the original algorithm, our solution, which we present in
Section~\ref{sec:strategy}, still works in polynomial
space. 
\ifreport
\else
\textbf{Additional lemmata and formal proofs of all results,
can be found in our technical report~\cite{GW:14}.}
\fi


\ifreport
\section{Preliminaries}
\label{sec:preliminaries}

We briefly introduce parity games in Section~\ref{sec:parity_games} and
Jurdzi\'nski's Small Progress Measures algorithm in Section~\ref{sec:spm}.
For an in-depth treatment of both, we refer to~\cite{2001automata} and the
references therein.

\fi

\ifreport
\subsection{Parity Games}
\label{sec:parity_games}
\else
\section{Parity Games}
\label{sec:parity_games}
\fi

A parity game is an infinite duration game, played by players \emph{odd},
denoted by $\odd$ and \emph{even}, denoted by $\even$, on a directed,
finite graph. The game is formally defined as follows.

\begin{definition}
A parity game is a tuple $(V, E, \priosym, (V_\even,V_\odd))$, where
\begin{itemize}
\item $V$ is a finite set of vertices, partitioned in a set $V_\even$ of
vertices owned by player $\even$, and a set of vertices $V_\odd$ 
owned by player $\odd$,
\item $E \subseteq V \times V$ is a total edge relation, \ie all vertices
have at least one outgoing edge,
\item $\priosym \oftype V \to \nat$ is a priority function that assigns
priorities to vertices.
\end{itemize}
\end{definition}
Parity games are depicted as graphs; diamond-shaped nodes
represent vertices owned by player $\even$, box-shaped nodes
represent vertices owned by player $\odd$ and the priority assigned
to a vertex is written inside the node, see the game depicted in
Figure~\ref{fig:example}.
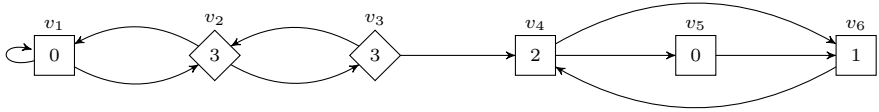
\begin{figure}[h!]
\centering
\begin{tikzpicture}[>=stealth']
\tikzstyle{every node}=[draw, inner sep=2pt, outer sep=0pt, node distance=40pt];
  \node[label=above:{\scriptsize $v_6$},shape=rectangle,minimum size=15] (y4)                    {\scriptsize 1};
  \node[label=above:{\scriptsize $v_5$},shape=rectangle,minimum size=15] (y5) [left of=y4,xshift=-20pt]                    {\scriptsize 0};
  \node[label=above:{\scriptsize $v_4$},shape=rectangle,minimum size=15]   (y6) [left of=y5,xshift=-20pt]                   {\scriptsize 2};
  \node[label=above:{\scriptsize $v_3$},shape=diamond,minimum size=19] (y7) [left of=y6,xshift=-20pt]                   {\scriptsize 3};
  \node[label=above:{\scriptsize $v_2$},shape=diamond,minimum size=19]   (y8) [left of=y7,xshift=-20pt]                   {\scriptsize 3};
  \node[label=above:{\scriptsize $v_1$},shape=rectangle,minimum size=15] (y9) [left of=y8,xshift=-20pt]                   {\scriptsize 0};

\path[->]
  (y4) edge[bend left] (y6)
  (y5) edge (y4)
  (y6) edge (y5)
  (y6) edge[bend left] (y4) 
  (y7) edge (y6) edge[bend right] (y8) 
  (y8) edge[bend right] (y7) edge[bend right] (y9)
  (y9) edge[loop left] (y9) edge[bend right] (y8) 
;
\pgfresetboundingbox
\path[use as bounding box] (-0.8\textwidth,-0.8) rectangle (0.27,0.8);

\end{tikzpicture}
\caption{A simple parity game in which 4 vertices are owned by player odd, 2 vertices
are owned by player even and with 4 different priorities.}
\label{fig:example}
\end{figure}
Throughout this section, assume that $G = (V,E,\priosym, (V_\even,
V_\odd))$ is an arbitrary parity game.  We write $v \to w$ iff
$(v,w) \in E$, and $\post{v}$ to denote the set of successors of
$v$, \ie\ $\{w \in V ~|~ v \to w\}$. For a set of vertices $W \subseteq V$, we will denote
the minimal priority occurring in $W$ with $\minprio{W}$; by 
$V_i$ we denote the set of vertices with priority $i$; likewise
for $V_{\le k}$. Henceforth, we assume that $\i$ denotes an arbitrary
player; that is $\i \in \{\odd,\even\}$. We write $\pnot{\i}$ for
$\i$'s opponent:  $\pnot{\even}=\odd$ and $\pnot{\odd}=\even$.  For
a set $A \subseteq V$, we define $G \cap A$ as the structure $(A,
E \cap (A \times A), \priosym|_A, (V_\even \cap A, V_\odd \cap A))$;
the structure $G \setminus A$ is defined as $G \cap (V \setminus
A)$. The structures $G \cap A$ and $G \setminus A$ are again a
parity game if their edge relations are total. 

\paragraph*{Plays and Strategies}
A sequence of vertices $v_1, \ldots, v_n$ is a \emph{path} if $v_m \to
v_{m+1}$ for all $1 \leq m < n$.  Infinite paths are defined in a similar
manner. We write $\pi_i$ to denote the $i^\textrm{th}$ vertex in a path $\pi$.

A game starts by placing a token on some vertex $v \in V$.  Players
$\even$ and $\odd$ move the token indefinitely according to a single
simple rule: if the token is on some vertex that belongs to player
$\i$, that player moves the token to an adjacent vertex.
An infinite sequence of vertices created this way is called a
\emph{play}.  The \emph{parity} of the \emph{least} priority that occurs
infinitely often on a play defines the \emph{winner} of the play:
player $\even$ wins if, and only if this priority is even. 


A \emph{strategy} for player $\i$ is a partial function $\sigma
\oftype V^+ \to V$ satisfying that whenever it is defined for a
finite path $u_1 \cdots u_n \in V^+$, both $u_n \in V_{\i}$ and
$\sigma(u_1\cdots u_n) \in \post{u_n}$. We denote the set of
strategies of player $\player$ by $\memstrategy{\player}$.  
An infinite play $u_1\
u_2\ u_3 \cdots$ is \emph{consistent} with a strategy $\sigma$ if
all prefixes $u_1 \cdots u_n$ of the play for which $\sigma(u_1\
\cdots u_n)$ is defined, $u_{n+1} = \sigma(u_1 \cdots u_n)$. The
set of all plays, consistent with some strategy $\sigma$, and starting
in $v$ is denoted $\plays{\sigma}{v}$. Some
strategy $\sigma$ is winning for player $\i$ from vertex $v$ iff
all plays consistent with $\sigma$ are won by player $\i$.  Vertex
$v$ is won by player $\i$ whenever she has a winning strategy for
all plays starting in vertex $v$.  Parity games are
\emph{determined}~\cite{EJ:91}, meaning that every vertex is won
by one of the players; they are even \emph{positionally determined},
meaning that if $\i$ wins a vertex then she has a winning
\emph{positional strategy}: a strategy that determines where to
move the token next based solely on the vertex on which the token
currently resides. Such strategies can be represented by a function
$\sigma\oftype V_{\i} \to V$, and the set of all such strategies for
player $\player$ is denoted $\strategy{\player}$.  \emph{Solving} a parity game $G$
essentially means computing the partition $(\winsubeven{G},\winsubodd{G})$
of $V$ into vertices won by player $\even$ and player $\odd$,
respectively.

\begin{example}
In the parity game of Figure~\ref{fig:example}, 
$v_1, v_2$ and $v_3$ are won by player $\even$; $v_4, v_5$ and $v_6$ are
won by player $\odd$. A winning positional strategy for player $\even$
is to play from $v_2$ to $v_1$ and from $v_3$ to $v_2$. A winning
strategy for $\odd$ is to move from $v_4$ to $v_6$ and
from $v_5$ to $v_6$.
\end{example}

\paragraph*{Attractors and Dominions}
An $\i$-\emph{attractor} into a
set $U \subseteq V$ contains all those vertices from which player $\i$ can
force any play to $U$; it is formally defined as follows.
\begin{definition} The $\i$-\emph{attractor} into a set $U \subseteq V$,
denoted $\attr{\i}{U}$, is the least set $A \subseteq V$ satisfying:
\begin{enumerate}
 \item $U \subseteq A$
 \item
 \begin{enumerate}
 \item if $w \in V_{\i}$ and $\post{w} \cap A \neq \emptyset$, then $w \in A$
 \item if $w \in V_{\pnot{\i}}$ and $\post{w} \subseteq A$, then $w \in A$
 \end{enumerate}
\end{enumerate}

\end{definition}
Observe that the complement of an attractor set of any subset of a
parity game induces a parity game, \ie $G  \setminus \attr{\i}{U}$
for any $U$ and $\i$ is a well-defined parity game.  Whenever we
refer to an \emph{attractor strategy} associated with $\attr{\i}{U}$,
we mean the positional strategy that player $\i$ can employ to force
play to $U$; such a strategy can be computed in $\mathcal{O}(|V|+|E|)$
using a straightforward fixpoint iteration.

A set of vertices $U$ is an $\i$-dominion whenever there is a
strategy $\sigma$ for player $\i$ such that every play starting in
$U$ and conforming to $\sigma$ is winning for $\i$ and stays within
$U$. A game is a \emph{paradise} for player $\i$ if the entire game
is an $\i$-dominion.

\ifreport
We shall frequently work with strategies or dominions \emph{in the context of a certain subgame $G' \subset G$}, which do not retain their properties when moving to a larger context of $G$. For instance, consider a subset of vertices $W \subset V$ that induces a subgame $G \cap W$, and moreover that there is a subset $D \subseteq W$ which is a \player-dominion \emph{in $G \cap W$}. Observe that, in general, $D$ is not a \player-dominion within $G$. In such cases we always explicitly state which context is assumed.

\fi



\newcommand{\Nat}{\ensuremath{\mathbb{N}}}
\newcommand{\prog}[3]{\ensuremath{\textsf{Prog}(#1,#2,#3)}}
\newcommand{\lift}[2]{\ensuremath{\textsf{Lift}(#1,#2)}}
\newcommand{\liftodd}[2]{\ensuremath{\textsf{Lift}_\odd(#1,#2)}}
\def\progname{\textsf{Prog}}
\def\liftname{\textsf{Lift}}

\ifreport
\subsection{Jurdzi\'nski's Small Progress Measures Algorithm}
\label{sec:spm}
\else
\section{Jurdzi\'nski's Small Progress Measures Algorithm}
\label{sec:spm}
\fi

The SPM algorithm works by computing a \emph{measure} associated with
each vertex that characterises even (resp.\ odd) cycles: it is such
that it decreases along the play with each ``bad'' odd priority
encountered, and only increases upon reaching a beneficial even
priority. In what follows, we recapitulate the essentials for defining
and studying the SPM algorithm; our presentation is---as in
the original paper by Jurdzi\'nski---from the perspective of player
$\even$. \medskip

Let $G = (V, E, \priosym, (V_{\even}, V_{\odd}))$ be a parity game.
Let $d$ be a positive number and let $\alpha
\in \Nat^d$ be a \emph{$d$-tuple} of natural numbers.  We number
its components from $0$ to $d-1$, \ie $\alpha = (\alpha_0, \alpha_1,
\dots, \alpha_{d-1})$, and let $<$ on such $d$-tuples be given by
the lexicographic ordering. These tuples will be used to (partially)
record how often we can or must see vertices of a particular priority
on all plays. We define a derived ordering $<_i$ on $k$-tuples and
$l$-tuples of natural numbers as follows:
\[
(\alpha_0, \alpha_1, \dots, \alpha_k) <_i (\beta_0,\beta_1,\dots,\beta_l)
\text{ iff }
(\alpha_0, \alpha_1, \dots, \alpha_i) < (\beta_0,\beta_1,\dots,\beta_i)
\]
where, if $i > k$ or $i >l$, the tuples are suffixed with $0$s.
Analogously, we write $\alpha =_i \beta$ to denote that $\alpha$ and
$\beta$ are identical up-to and including position $i$.
Intuitively, the derived ordering provides enough information to
reason about the interesting bits of plays: when encountering a
priority $i$ in a play, we are only interested in how often we can
or must visit vertices of a more significant (\ie smaller) priority
than $i$,  whereas we no longer care about the less significant priorities that we shall encounter.

Now, assume from hereon that $d-1$ is the largest priority occurring in
$G$; \ie, $d$ is one larger than the largest priority
in $G$.  For $i \in \Nat$, let $n_i = |V_i|$. 
Define $\mathbb{M}^{\even} \subseteq \Nat^d \cup \{\top\}$, containing $\top$
($\top \notin \Nat^d$)
and only those $d$-tuples with $0$ on \emph{even} positions and
natural numbers $\leq$ $n_i$ on \emph{odd} positions $i$.

The lexicographical ordering $<$ and the family of orderings
$<_i$ on $d$-tuples is extended to an ordering on $\mathbb{M}^\even$
by setting $\alpha < \top$ and $\top = \top$. 
Let $\rho {:} V \to \mathbb{M}^\even$ and suppose $w \in \post{v}$. Then
$\prog{\rho}{v}{w}$ is the least $m \in \mathbb{M}^{\even}$, such that
\begin{itemize}
  \item $m \geq_{\prio{v}} \rho(w)$ if $\prio{v}$ is even,
  \item $m >_{\prio{v}} \rho(w)$, or $m = \rho(w) = \top$ if $\prio{v}$ is odd.
\end{itemize}
\begin{definition} Function
$\rho$ is a \emph{game parity progress measure}
if for all $v \in V$:
\begin{itemize}
  \item if $v \in V_{\even}$, then for some $w \in \post{v}$, $\rho(v) \geq_{\prio{v}} \prog{\rho}{v}{w}$
  \item if $v \in V_{\odd}$, then for all $w \in \post{v}$, $\rho(v) \geq_{\prio{v}}\prog{\rho}{v}{w}$
\end{itemize}
\end{definition}

\begin{proposition}[Jurdzi\'nski~\cite{Jur:00}]\label{prop:jurdzinski}
If $\rho$ is the \emph{least} game parity progress measure, then for all $v \in V$:
$\rho(v) \neq \top$ iff $v \in W_{\even}$.
\end{proposition}
The least game parity progress measure can be characterised as the least fixpoint
of a monotone transformer on the complete lattice we define next. Let
$\rho,\rho' {:} V \to \mathbb{M}^\even$ and define
$\rho \sqsubseteq \rho'$ as $\rho(v) \leq \rho'(v)$ for all $v \in V$. We write
$\rho \sqsubset \rho'$ if $\rho \sqsubseteq \rho'$ and $\rho \not= \rho'$. Then
the set of all functions $V \to \mathbb{M}^\even$ with $\sqsubseteq$ is a complete
lattice. The monotone transformer defined on this set is given as follows:
\[
\lift{\rho}{v} = \begin{cases}
\rho[v \mapsto \max \{ \rho(v), \min\{ \prog{\rho}{v}{w} \mid v \to w \} \}] &
\text{if $v \in V_{\even}$}\\
\rho[v \mapsto \max \{ \rho(v), \max\{ \prog{\rho}{v}{w} \mid v \to w \} \}] &
\text{if $v \in V_{\odd}$}
\end{cases}
\]
As a consequence of Tarski's fixpoint theorem, we know the least fixpoint of the
above monotone transformer exists and can be computed using Knaster-Tarski's iteration
scheme. This leads to the original SPM algorithm, see Algorithm~\ref{alg:spm}.
\begin{algorithm}[h!t]
\begin{algorithmic}[1]
\Function{SPM}{$G$}
\State \emph{\textbf{Input} $G = (V, E, \priosym, (V_\even, V_\odd))$}
\State \emph{\textbf{Output} Winning partition $(\winsubeven{G},\winsubodd{G})$}
\State $\rho  \gets \lambda v \in V.~(0, \dots, 0)$
\While{$\rho \sqsubset \lift{\rho}{v}$ for some $v \in V$}
  \State $\rho \gets \lift{\rho}{v}$ for some $v \in V$ such that $\rho \sqsubset \lift{\rho}{v}$
\EndWhile
\State \Return $(\{v \in V ~|~ \rho(v) \not= \top\}, \{v \in V ~|~ \rho(v) = \top\})$
\EndFunction
\end{algorithmic}
\caption{The original Small Progress Measures Algorithm}
\label{alg:spm}
\end{algorithm}
Upon termination of the iteration within the SPM algorithm, the
computed game parity progress measure $\rho$ is used to compute the
sets $\winsubeven{G}$ and $\winsubodd{G}$ of vertices won by player $\even$ and
$\odd$, respectively. 
\begin{theorem}[See~\cite{Jur:00}] Algorithm~\ref{alg:spm} solves a parity
game in $\runtimefloor$.
\end{theorem}
The above runtime complexity is obtained by considering the more optimal
runtime of solving a game $G$, or $G$'s 'dual', obtained by shifting
all priorities by one and swapping ownership of all vertices.
The runtime complexity for computing winning strategies for both players using
the SPM algorithm is worse.  A winning strategy
$\sigma_\even {:} V_\even \to V$ for player $\even$ can be extracted
from $\rho$ by setting $\sigma_\even(v) = w$ for $v \in V_\even
\cap \winsubeven{G}$ and $w \in \post{v}$ such that $\rho(w) \le \rho(w')$
for all $w' \in \post{v}$.  A winning strategy for player $\odd$
cannot be extracted from $\rho$ \emph{a posteriori}, so, as also
observed in~\cite{Sch:07}, the runtime
complexity of computing a winning strategy cannot be improved by
considering the dual of a game (contrary to what is stated in~\cite{Jur:00}). 
\begin{theorem}[See also~\cite{Sch:07}]
Algorithm~\ref{alg:spm} can compute
winning strategies for both players in~$\runtimeceil$. 
\end{theorem}

To facilitate the analysis of SPM, we will use the following  terms and notions. The term \emph{measure} refers to the intermediate values of $\rho$ in the lifting process as well. Given a tuple $m \in \mathbb{M}^\even$, we say that the position $i$ in $m$ is \emph{saturated}, if $(m)_i = |V_i|$.

\ifreport
\def\liftseq{ms}

A convenient abstraction of an instance of SPM being executed on a particular game is a sequence of intermediate measure values $\rho_0 \rho_1 \dots \rho_C$, where $\rho_C$ is the current measure value (as we frequently consider partial executions of SPM, $\rho_C$ is not necessarily the final, stable measure). Formally, we define a \emph{lifting context} as a tuple $\tuple{G,\liftseq}$, where $G$ is a parity game, and 
$\liftseq = \rho_0 \rho_1 \dots \rho_C$ a sequence of all intermediate measure values.
\fi

\section{An operational interpretation of progress measures}
\label{sec:interpretation}

\def\globmaxprio{max \priosym}

\def\eventups{\mathbb{M}^{\even}}
\def\exteventups{\mathbb{M}^{\even}_{ext}}
\def\oddtups{\mathbb{M}^{\odd}}

\def\allplays{\Pi}
\def\profilesym{\theta}
\newcommand{\profilefun}[1]{\profilesym_{#1}}
\newcommand{\profile}[2]{\profilesym_{#1}(#2)}

\newcommand{\maxval}[2]{\varphi^{*}_{#1}(#2)}

\newcommand{\succtup}[2]{\textsf{succ}_{#1}(#2)}

While SPM is a relatively simple algorithm in the sense that it is
concise and its individual steps are elementary operations, it lacks
a clear and appealing explanation of the devices used. It is therefore
difficult to understand, and possibly enhance. Apart from the formal
definition of progress measures, little explanation
of what is hidden behind the values in tuples is offered.
Notable exceptions are \cite{Klau:02}, which explains that when restricted
to $\odd$-solitaire games, one can use the maximal degrees of `odd
stretches' (a concept we make precise below) in order to define a
certain parity progress measure, and Schewe's bigstep paper \cite{Sch:07},
where it is shown that dominions of a bounded size can be detected
using measures with a restricted codomain. In general, the  high-level
intuition is that the larger progress measure values indicate more
capabilities of player $\odd$, and a value at a given position in the
tuple is somehow related to the number of priorities that $\odd$
can enforce to visit.

In what follows, we present a precise and operational interpretation
of measures. Our interpretation is similar in spirit to the one used in \cite{Klau:02},
but applicable to all parity games, and uses a concept known
from the realm of strategy improvement algorithms -- a value (or
profile) of a play. Here, values are defined in terms of maximal odd-dominated
stretches occurring in a prefix of a play. With this notion at hand,
we can consider an optimal valuation of vertices, being the best
lower bound on play values that player $\even$ can enforce, or the
worst upper bound that $\odd$ can achieve, \ie it is an \emph{equilibrium}.
The optimal valuations range over the same codomain as progress
measures, and the main result of this section states that the least
game parity progress measure is equal to the optimal valuation.


Let $\exteventups$ denote all tuples in $\nat^{d} \cup \{\top\}$
such that for all $m \in \exteventups$ and even positions $i \leq
d$, $(m)_i =0$; \ie compared to $\eventups$, we drop the requirement that
values on odd positions $i$ are bounded by $|V_i|$. Elements in $\exteventups$ are ordered in the same fashion as those in $\eventups$. Given a 
play $\pi$, a \emph{stretch} is a subsequence
$\pi_s \pi_{s+1} \dots \pi_{s+l}$ of $\pi$. For a priority $k$, a
\emph{$k$-dominated stretch} is a stretch in which the minimal
priority among all vertices in the stretch is $k$. The \emph{degree}
of a $k$-dominated stretch is the number of vertices with priority
$k$ occurring in the stretch.
\begin{definition}
Let us denote all plays in the parity game by $\allplays$. An $\even-$\emph{value} (or simply value) of a play is a function $\profilefun{\even}: \allplays \longrightarrow \exteventups$ defined as follows:
\begin{itemize}
 \item if $\pi$ is winning for $\odd$, then $\profile{\even}{\pi} = \top$
 \item if $\pi$ is winning for $\even$, then $\profile{\even}{\pi} = m$, where 
 $m \neq \top$, and for every odd position $i$, $(m)_i$ is the \emph{degree} of the maximal $i$-dominated stretch that is a prefix of $\pi$ 
\end{itemize}
\end{definition}
Observe that a play value is well-defined, as an
infinite $i$-dominated stretch for an odd $i$ implies that a game
is won by $\odd$, and its value is $\top$ in such case. 

\begin{example}
Suppose that the sequence of priorities corresponding to a certain
play $\pi$ is $453453213(47)^{*}$. Then $\profile{\even}{\pi} =
(0,1,0,2,0,0,0,0)$.

\end{example}
 
 

\ifreport

\paragraph{Successor up-to-$k$} For $m \in \eventups \setminus \{\top\}$ and $k \in \nat$, we will denote with $\succtup{k}{m}$ the least $m' \in \eventups$ such that $m' >_{k} m$.

\begin{lemma}
\label{lem:evencan}
 If $\rho$ is a game progress measure of a parity game $G$, then for all $v$ there is a strategy $\sigma_{\even} \in \strategy{\even}$ such that for every $\pi \in \plays{\sigma_{\even}}{v}$, $\profile{\even}{\pi} \leq \rho(v)$
\end{lemma}

\begin{proof}
 We focus on the nontrivial case when $\rho(v) \neq \top$; we will show that player $\even$ has a strategy to force plays with values not exceeding $\rho(v)$.  The strategy in question (denoted by $\sigma_{min}$) is the same as used by the SPM algorithm -- $\even$ always picks the vertex minimising $\rho(v)$.
 
 We proceed with induction on $\rho(v)$. To prove the base case $(\rho(v) = (0,0,\dots,0))$, observe first that if $\pi \in \plays{\sigma_{min}}{v}$, then for all $k$ such that $\rho(\pi_k) <_i \rho(\pi_{k+1})$, we have $\prio{\pi_k} < i$ (this follows from the definition of $\sigma_{min}$ and properties of game progress measures).
 
 Take any $\pi \in \plays{\sigma_{min}}{v}$. Suppose, towards contradiction, that 
 $(\profile{\even}{\pi})_j > 0$ for some odd $j$. Let $\pi_l$ be the first vertex with priority $j$ in $\pi$. From the definition of game progress measure, $\rho(\pi_l) >_j (0,0,\dots,0)$. Since $\rho(\pi_0) = (0,0,\dots,0)$, this means that for some $k$, $0 \leq k \leq l-1$, we must have $\rho(\pi_{k}) <_j \rho(\pi_{k+1})$; and from the inital observation we obtain $\prio{\pi_k}<j$. But it means that a value smaller than $j$ occurs before the first occurence of priority $j$ in $\pi$, hence  $(\profile{\even}{\pi})_j=0$, a contradiction.
 
  For the inductive step, we assume that whenever the value of $\rho(w)$ for any game progress measure of an arbitrary game $G$ and its vertex $w$ is lower than $m$, then for all plays consistent with $\sigma_{min}$ and starting at $w$, their values do not exceed $\rho(w)$.
 
 Take $v$ with $\rho(v)=m$. Let $\pi = v_0 v_1 \dots$ be a play starting at $v_0=v$ and conforming to $\sigma_{min}$, and let $m' = \profile{\even}{\pi}$. Observe that since $\sigma_{min}$ is a memoryless strategy winning for player $\even$ \cite{Jur:00}, $\pi$ cannot pass any odd-dominated cycles, and we have $\profile{\even}{\pi} \in \eventups$. 
 
 We proceed to prove that $m' \leq  m$. Let $k$ be the largest (least significant) position such that $(m')_k > 0$. There exists a non-trivial $k$-dominated stretch in the prefix of $\pi$; let $v_n$ be the first vertex with priority $k$ occurring in $\pi$. From the game progress measure property, the way $\sigma_{min}$ is defined, and the fact that $k$ is the least priority occurring between $v_1$ until $v_n$, we know that $\rho(v_i) \geq_{k} \rho(v_{i+1})$ for all $i: 0 \leq i \leq n$, and moreover for $i=n$ the inequality is strict. Hence we have $m = \rho(v) \geq_{k} \rho(v_n) >_{k} \rho(v_{n+1})$. 
 
 Let $\pi_{post}$ be the postfix of $\pi$ starting with $v_{n+1}$. By applying the inductive hypothesis to $v_{n+1}$, we know that $\profile{\even}{\pi_{post}} \leq \rho(v_{n+1})$, and hence $\profile{\even}{\pi_{post}}<_{k} m$.  Since $v_n$ was the first occurrence of priority $k$ in $\pi$, and dominating the prefix, we have $(\profile{\even}{\pi})_k = (\profile{\even}{\pi_{post}})_k + 1$, and $(\profile{\even}{\pi})_i = (\profile{\even}{\pi_{post}})_i$ for $i<k$. In short, we have thus $m' = \profile{\even}{\pi} = \succtup{k}{\profile{\even}{\pi_{post}}}$ (here, we use the fact that $\profile{\even}{\pi} \in \eventups$). Since $m >_{k} \profile{\even}{\pi_{post}}$, we obtain $m \geq \succtup{k}{\profile{\even}{\pi_{post}}} = m'$. \qed 
\end{proof}

\begin{lemma}
\label{lem:oddcan}
 If $\overline{\rho}$ is the least game progress measure of a parity game $G$, then there is a strategy $\sigma_{\odd} \in \memstrategy{\odd}$ such that for every $\pi \in \plays{\sigma_{\odd}}{v}$, $\profile{\even}{\pi} \geq \overline{\rho}(v)$
\end{lemma}

\begin{proof}

The strategy in question is the same as the lifting-history based strategy that we introduce in section \ref{sec:bounded}; the lemma follows directly from Proposition \ref{prop:oddresponse_playvalue}.

\end{proof}

From lemmata \ref{lem:evencan} and \ref{lem:oddcan} we obtain the following theorem.

\fi

\begin{theorem}
\label{thm:progint}
If $\overline{\rho}$ is the least progress measure of a parity game $G$, then, for all $v$:
\begin{enumerate}
 \item there is a strategy $\sigma_{\odd} \in \memstrategy{\odd}$ such that for every $\pi \in \plays{\sigma_{\odd}}{v}$, $\profile{\even}{\pi} \geq \overline{\rho}(v)$
 \item there is a strategy $\sigma_{\even} \in \strategy{\even}$ such that for every $\pi \in \plays{\sigma_{\even}}{v}$, $\profile{\even}{\pi} \leq \overline{\rho}(v)$
\end{enumerate}
\end{theorem}

The above theorem links the progress measure values to players' capabilities to enforce beneficial plays or avoid harmful ones, where the benefit from a play is measured by a specially devised play value, as it is done in strategy improvement algorithms. 
This offers a more operational view on progress measure values, which, combined with a more fine-grained analysis of the mechanics of SPM allows us to extract winning strategies for both players in the next section.


A notable difference between strategy improvement algorithms and SPM is
that SPM does not work with explicit strategies, and the intermediate
measure values do not represent any proper valuation induced by
strategies -- only the final least progress measure does. Still,
these intermediate values give an underapproximation of the
capabilities of player $\odd$ in terms of odd-dominated stretches
that she can enforce.

Note that a consequence of Theorem \ref{thm:progint} is that the least (resp. greatest) play values that player $\odd$ (resp. $\even$) can enforce are equal, and coincide with $\overline{\rho}$.

\section{Strategy construction for player \odd}
\label{sec:strategy}

Computing winning strategies is typically part of a practical
solution to a complex verification or a controller synthesis problem.
In such use cases, these strategies are employed to
construct witnesses and counterexamples for the verification problems,
or for constructing control strategies for the controller~\cite{AVW:03}.
As we explained in Section \ref{sec:spm}, the SPM algorithm can
easily be extended to construct a winning strategy for player
$\even$. The problem of deriving a winning strategy for player \odd
in SPM (other than by running the algorithm on the `dual' game, or
by using a `dual' domain $\mathbb{M}^{\odd}$) has, however, never
been addressed. Note that the problem of computing strategies is
at least as hard as solving a game. Indeed, even if we are equipped
with a valid winning partition $(\winsubeven{G},\winsubodd{G})$ for
a game $G$, then deriving the strategies witnessing these partitions
involves the same computational overhead as the one required to
compute $(\winsubeven{G},\winsubodd{G})$ in the first place.
\begin{restatable}{proposition}{propstrategiesdifficult} 
\label{prop:strategies_difficult}
The problem of finding the winning partition
$(\winsubeven{G},\winsubodd{G})$ of a given game $G$ can be reduced
in polynomial time to the problem of computing a winning strategy
for player $\player$ in a given $\player$-dominion.
\end{restatable}

\ifreport
\begin{proof}
 We will reduce the problem of recognising whether a given set $D$ is a dominion of a given player to the strategy derivation problem. The former problem is known to be polynomially equivalent to the winning partition problem \cite{DKT:2012}.

 Suppose there is an algorithm ${\cal A}$ that, given a dominion $D \subseteq V(G)$ of player $\player$, computes a winning strategy $\sigma$ of player $\player$, closed on $D$. Moreover, we assume that the worst-case running time of ${\cal A}$ has an upper bound $T(|V|,|E|,d)$. We can construct an algorithm ${\cal A}'$ that decides whether $D$ is a $\player$-dominion in $O(T(|V|,|E|,d) + (|V|+|E|) \cdot \log d)$ by simply running ${\cal A}$ on $D$ and analysing the outcome.
 \begin{itemize}
  \item ${\cal A}$ has not returned a well-defined strategy $\sigma$ within $T(|V|,|E|,d)$ steps. In this case the answer is \textbf{no}
  \item ${\cal A}$ has returned some answer $\sigma$ within $T(|V|,|E|,d)$ steps. By solving the induced solitaire game in $(|V|+|E|) \cdot \log d$ time, we verify whether $\sigma$ is indeed a winning strategy for $\player$ on $D$. Is so, return \textbf{yes}, otherwise return \textbf{no}. 
 \end{itemize} 
\end{proof}

\fi

\begin{figure}[h]
\centering
\begin{tikzpicture}[>=stealth']
\tikzstyle{every node}=[draw, inner sep=0pt, outer sep=0pt, node distance=30pt];
  \node[shape=rectangle,minimum size=23] (y1)                            {\scriptsize $2$};
  \node[shape=rectangle,minimum size=23] (y2) [right of=y1,xshift=20pt]  {\scriptsize $1$};
  \node[shape=rectangle,minimum size=23] (y21) [right of=y2,xshift=20pt]  {\scriptsize $3$};
  \node[draw=none] (y3) [right of=y21,xshift=20pt] {\Huge $\dots$};
  \node[shape=rectangle,minimum size=23] (y4) [right of=y3,xshift=20pt]  {\scriptsize $2N-2$};
  \node[shape=rectangle,minimum size=23] (y5) [right of=y4,xshift=20pt]  {\scriptsize $2N-1$};
  \node[shape=rectangle,minimum size=23] (y6) [right of=y5,xshift=20pt]  {\scriptsize $2N$};

\path[->]
  (y1) edge (y2) 
  (y2) edge (y21)
  (y21) edge (y3)
  (y3) edge (y4)
  (y4) edge (y5)
  (y5) edge (y6)
  (y6.south west) edge[bend left=7] (y1.south east) 
  (y6) edge[loop right] (y6)
;
\end{tikzpicture}
\caption{A parity game won by player $\odd$. Solving the game using
$\mathbb{M}^{\even}$, the first $\top$ value is reached after the first full pass
of the cycle containing priority $1$ ($O(N^2)$ using a fair lifting
strategy), and it will then propagate through the game.
Solving the dual game, or using $\mathbb{M}^{\odd}$ takes exponential
time to lift the node with priority $2N$.}

\label{fig:tops_faster}
\end{figure}
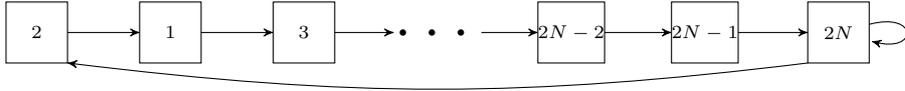
Deriving a strategy for both players by using the SPM to compute
$\mathbb{M}^{\even}$ measures and $\mathbb{M}^{\odd}$ measures
consecutively, or even simultaneously, affects, as we already
discussed in Section~\ref{sec:spm}, SPM's runtime.  This is nicely
illustrated by the family of games depicted in
Figure~\ref{fig:tops_faster}, for which lifting to top using an
even-biased measure is exponentially faster than arriving at a
stable ``non-top'' odd-biased measure.  Being able to compute \odd
strategies without resorting to the aforementioned methods would
allow us to potentially significantly improve efficiency on such
instances.  It may be clear, though, that extracting a winning
strategy from the small progress measures algorithm for vertices
with measure $\top$ will require modifying the algorithm (storing
additional data, augmenting the lifting process).  For instance,
simply following the vertex that caused the update to top, fails, as
the example below shows.

\begin{example} \label{ex:greedytop_wrong} Reconsider the game
depicted in Figure~\ref{fig:example}. Recall that 
vertices $v_4, v_5$ and $v_6$ are won by player $\odd$, and in all
possible lifting schemes, the first vertex whose measure becomes
$\top$ is $v_6$. After that, a possible sequence of liftings is
that first $\rho(v_5)$ is set to $\top$ (due to $v_6$), followed
by $\rho(v_4) = \top$ (due to $v_5$). If we set the strategy
based on the vertex that caused the given vertex to be lifted to top,
we obtain $\sigma(v_4) = v_5$, which is not winning for
$\odd$.

\end{example}
The key problem is that for vertices won by player \odd, from some
point onward, the lifting process discards significant information.
This is best seen in case of lifting to $\top$ -- a partial
characterisation of reachable odd priorities contained in a tuple
(see also our previous section) is ultimately replaced with a mere
indication that player \odd can win.

\ifreport
\subsection{Key observation}
\label{sec:intuition}

At this point we shall give an intuitive explanation of the main insight that enables us to define part of player \odd strategy in the course of lifting, once a top value is reached. In section \ref{sec:bounded} the observations made here will be formalised and proved, leading to Theorem \ref{thm:centraal}, which forms the basis of our algorithm.

\else
\subsection{A Bounded $\odd$-Dominion}
\fi

Consider a game $G$ on which a standard SPM algorithm with an arbitrary lifting policy has been applied. Suppose that at some point a vertex $v$ is the first one to be lifted to $\top$, and after lifting of $v$ the algorithm is suspended, resulting is some temporary measure $\rho$. Let $k$ be the priority of $v$. 

We will start with two straightforward observations. The first one
is that $k$ must be an odd number; this is because a vertex with an
even priority obtains, after lifting, a $\rho$-value equal to the
$\rho$-value of one of its successors, and therefore it cannot be
the first vertex lifted to $\top$. Another immediate conclusion is
that at least one (or all, if $v \in V_{\even}$) successor(s) of
$v$ has (have)  a $\rho$-value saturated up to the $k$-th position,
\ie it is of the form $m = (0,|V_1|, 0, |V_3|,\dots, 0, |V_k|, ***)$;
were it not the case, then a non-top value $m'$ such that $m' >_{k}
m$ would exist, which would be inconsistent with the definition of
\progname.

\begin{figure}[h!]
\centering
\begin{tikzpicture}[>=stealth']
\tikzstyle{every node}=[draw, inner sep=2pt, outer sep=0pt, node distance=40pt];

\node[draw=none] (x) {};
\node[draw=none] (w) [below of=x,xshift=-10pt,yshift=37pt] {};
\path[draw,fill=gray!8] (w) ellipse (2.4cm and 1.7cm) ;
\node[shape=rectangle,minimum size=12pt,label=left:{\scriptsize $v$},label=right:{\scriptsize $\rho(v) = \top$}] (y) [above of=x,xshift=-28pt,yshift=-15pt] {\scriptsize $k$};

\node[shape=rectangle,minimum size=12pt,label=left:{\scriptsize $u_{\text{max}}$}] (z1) [below of= y,xshift=-20pt] {};

\node[draw=none] (z2) [below of=y,xshift=0pt] {$\dots$};
\node[draw=none] (z3) [below of=y,xshift=20pt] {};

\node[draw=none] (z5) [below of=y,xshift=92pt] {};
\node[draw=none] (z6) [below of=y,xshift=104pt] {};
\node[draw=none] (u) [below of=z2,yshift=20pt] {\scriptsize \begin{tabular}{l}$\odd$-dominion \\ $D \subseteq \bigcup_{i \ge k} V_i$ \end{tabular} };

\path[draw] (x) node[draw=none,xshift=63pt,yshift=10pt] {\Large $V$} ellipse (2.8cm and 2.0 cm);

\draw[->] (y) edge[color=red] (z1) edge  (z3) edge (z5) edge (z6)
;

\pgfresetboundingbox
\path[use as bounding box] (-2.8,-2) rectangle (2.8,2.0);
\end{tikzpicture}
\qquad\quad
\begin{tikzpicture}[>=stealth']
\tikzstyle{every node}=[draw, inner sep=2pt, outer sep=0pt, node distance=40pt];

\node[draw=none] (x) {};
\node[draw=none] (w) [below of=x,xshift=-10pt,yshift=37pt] {};
\path[draw,fill=gray!8] (w) ellipse (2.4cm and 1.7cm) ;
\node[shape=diamond,minimum size=17pt,label=left:{\scriptsize $v$},label=right:{\scriptsize $\rho(v) = \top$}] (y) [above of=x,xshift=-28pt,yshift=-15pt] {\scriptsize $k$};
\node[draw=none] (z1) [below of=y,xshift=-20pt] {};
\node[draw=none] (z2) [below of=y,xshift=0pt] {$\dots$};
\node[draw=none] (z3) [below of=y,xshift=20pt] {};
\node[draw=none] (u) [below of=z2,yshift=20pt] {\scriptsize \begin{tabular}{l}$\odd$-dominion \\ $D \subseteq \bigcup_{i \ge k} V_i$ \end{tabular} };

\path[draw] (x) node[draw=none,xshift=63pt,yshift=10pt] {\Large $V$} ellipse (2.8cm and 2.0 cm);

\draw[->] (y) edge (z1) edge  (z3)
;

\pgfresetboundingbox
\path[use as bounding box] (-2.8,-2) rectangle (2.8,2.0);
\end{tikzpicture}
\caption{Snapshot of the SPM algorithm after the first vertex $v$ is lifted to top.}
\label{fig:egg}
\end{figure}
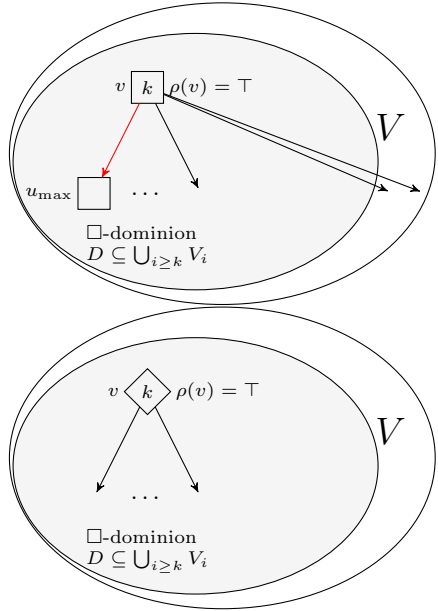
There are two more complex properties, which we can utilise to modify the SPM algorithm and compute the winning strategy for player \odd (see Figure \ref{fig:egg}). 
\begin{enumerate}
 \item Vertex $v$ belongs to an \odd-dominion $D$ within $G$ such that the minimal priority in $D$ is $k$.
 \item If $v \in V_{\odd}$, then picking the successor $u_{max}$ of $v$ with the maximal current $\rho$-value among $\post{v}$ is a part of a (positional) winning strategy for \odd that stays within such a dominion $D$ as described above. 
 \end{enumerate}
The intuition concerning the above facts is as follows. Vertices with
a measure value $m$ saturated up to but possibly excluding a certain
position $i$ have a neat interpretation of the measure value at
position $i$: \medskip

\noindent\emph{Player \odd can force the following outcome of a play:
\begin{enumerate}
 \item priority $i$ appears $m_i$ times without any lower priority in between
 \item it will reach a $\top$-labelled vertex via priorities not more significant than $i$
 \item it enters a cycle with an odd dominating priority larger (less significant) than $i$.
\end{enumerate}
}\medskip

\noindent
Therefore, in the situation as described above, \odd can force a
play starting at $v$ to first go in one step to the successor
$u_{max}$ of $v$ with a measure of the form $(0,|V_1|, 0, |V_3|,
\dots, 0, |V_k|, ***)$, and then to play further and either force
a less significant odd-dominated cycle (cases 2 and 3, since $v$
is the only $\top$-labelled vertex), or to visit vertices with
priority $k$ $|V_k|$ times without any lower priority in between.
But in the latter case, since $v$ has priority $k$, we have in fact
$|V_k| + 1$ vertices with priority $k$ not ``cancelled'' by a lower
priority. Hence player \odd has forced an odd-dominated cycle with
the lowest (most significant) priority $k$.  Note that this does
not imply we can simply construct a winning strategy for $\odd$ by
always picking a successor with the maximal measure to further
vertices that can be visited along the play; such a method may lead
to an erroneous strategy, as illustrated by Figure \ref{fig:greedywrong}.

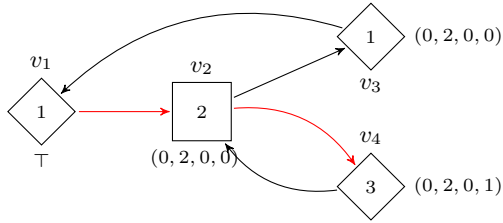
\begin{figure}[h!]
\centering
\begin{tikzpicture}[>=stealth']
\tikzstyle{every node}=[draw, inner sep=1pt, outer sep=1pt, node distance=40pt];
  \node[shape=diamond,minimum size=26,label=above:{$v_1$},label=below:{\scriptsize $\top$}] (y1)                            {\scriptsize 1};
  \node[shape=rectangle,minimum size=22,label=above:{$v_2$},label=below:{\scriptsize $(0,2,0,0)$~~~}] (y2) [right of=y1,xshift=20pt]  {\scriptsize 2};
  \node[shape=diamond,minimum size=26,label=below:{$v_3$},label=right:{\scriptsize $(0,2,0,0)$}] (y3) [above right of=y2,xshift=35pt]  {\scriptsize 1};
  \node[shape=diamond,minimum size=26,label=above:{$v_4$},label=right:{\scriptsize $(0,2,0,1)$}] (y4) [below right of=y2,xshift=35pt]  {\scriptsize 3};

\path[->]
  (y1) edge[color=red] (y2) 
  (y2) edge (y3) edge[bend left,color=red] (y4)
  (y3) edge[bend right] (y1)
  (y4) edge[bend left] (y2)
;
\pgfresetboundingbox
\path[use as bounding box] (-2,-1.4) rectangle (0.5*\textwidth,1.4);
\end{tikzpicture}
\caption{A game, won entirely by player $\odd$, and demonstrating
that a strategy defined by a greedy choice of vertex with the maximal 
tuple does not work. After lifting the vertices in order $v_1,v_3,v_2,v_4,v_1$, we obtain the measure values as above. Player \odd would then choose $v_2$, which leads to a losing play, whereas the choice of the other successor $(v_1)$ yields a winning play for \odd.}
\label{fig:greedywrong}
\end{figure}	

\ifreport
\else
Propagating a top value only to vertices with less significant
priorities is, however, safe. This can be achieved efficiently by
a slightly modified attractor that works within a given context of vertices $W$, 
which we call a \emph{guarded} attractor.
\begin{definition}
Let $k$ be some priority and $U,W$ some sets for which
$U \subseteq W \cap V_{\geq k}$. 
Then $\myattrW{\ge k}{\odd}{U}{W}$ is the least set $A$ satisfying:
\begin{enumerate}
 \item $ U \subseteq A \subseteq W \cap V_{\geq k}$
 \item 
 \begin{enumerate}
  \item if $u \in V_{\odd}$ and $\post{u} \cap A \neq \emptyset$, then $u \in A$
  \item if $u \in V_{\even}$ and $\post{u}\cap W \subseteq A$, then $u \in A$
 \end{enumerate}
\end{enumerate}
If $W = V$, we drop the subscript $W$ from the guarded attractor.
\end{definition}

\fi

\newcommand{\exttuples}[1]{\mathbb{M}^{#1,\textsf{ext}}}
\def\extprog{newProg}

\ifreport
\section{A bounded \odd dominion}
\label{sec:bounded}

In this section, we formally prove the key insight concerning the bounded $\odd$ dominion and the partial strategy construction, that will be later used in the strategy derivation algorithm. 

\subsection{Lifting History Graph}
\label{sec:lhg}

We first introduce an auxiliary notion of a Lifting History Graph. Its nodes (states) contain all snapshots of $\rho$-values that appeared at every parity game vertex in the course of the lifting, whereas the edges explain the causal dependency between $\rho$-values at a given vertex and its successors. In other words, the graph contains the entire history of lifting up to a certain point, and along its edges we can ``move back'' in the history of updates.

\begin{definition}{(Lifting History Graph)}
Suppose we are in the context of some partial execution of the SPM algorithm, in which $t$ liftings have been performed on a certain parity game $G=(V,E,\priosym,(V_{\even},V_{\odd}))$, starting with {\mbox{ $\rho_0 = \lambda w \in V.~(0, \dots, 0)$}} and yielding after each $i$-th lifting a temporary measure $\rho_i$. We define the corresponding \emph{Lifting History Graph} $LH=(V_{LH},E_{LH})$.
The set of nodes $V_{LH} \subseteq V \times \mathbb{M}^{\even}$ contains all pairs $(v,m)$ such that at some stage of lifting $v$ had a value $m$ (i.e. there is $i$ such that $\rho_i(v) = m$), and we define the edge relation as: 
\[
\begin{array}{lll}
\post{(v,m)} & \,\eqdef\, & \{(w,m') \mid w \in \post{v} \text{(in $G$)} \,\wedge\, \exists i \leq t: m = \rho_{i}(v) > \rho_{i-1}(v)\\
& & \text{and either } v \neq w \,\wedge\, \rho_{i-1}(w) = \rho_{i}(w) = m'\\
& & \text{or } v=w \,\wedge\, \rho_{i-1}(v) = m'\} 
\end{array}
\]
that is, the successors of $(v,m)$ in $LH$ are those pairs $(w,m')$ such that $w \in \post{v}$ and when $v$ was lifted to $m$, $\rho(w)$ had value $m'$. In other words,  $\post{(v,m)}$ constitutes a ``snapshot'' of $\rho$-values of $v$'s successors just before $v$ was lifted to $m$.
\end{definition}

The following technical proposition summarises how the $\rho-$values in the Lifting History Graph change as we move one step back in the history of depenendencies.

\def\mmin{m^{\min}}
\def\mmax{m^{\max}}
\def\plh{PLH}

\begin{proposition}
\label{prop:lhg-properties}
 Let $LH=(V_{LH},E_{LH})$ be a lifting history graph, and $(w,m) \in V_{LH}$ a position in $LH$ such that  $\post{(w,m)} \neq \emptyset$. Let us denote $\mmin$ and $\mmax$ respectively the minimal and maximal value of the set $\{m' \,\mid\, (u,m') \in \post{(w,m)}\}$.  
 
 \begin{enumerate}
  \item if $m \neq \top$, then 
  \begin{center}
  for all $i > \prio{w}$, $(m)_{i} = 0$ ~\hspace{5pt}(\plh.0) 
  \end{center} 
  
  and one of the following holds:
 \[ 
 \begin{array}{|l|l|lr|}
 \hline
 w \in V_{\even} ~&~ \prio{w} \text{ even} ~&~ \mmin =_{\prio{w}} m & ~\hspace{5pt}\text{(\plh.11)} \\
 \hline 
 w \in V_{\even} ~&~ \prio{w} \text{ odd} ~&~ \mmin =_{i-1} m \,\,\wedge\,\, (\mmin)_{i}=(m)_{i}\!-1 & \text{(\plh.12)}\\
 & & \wedge \text{ for all $j \in \{i+1, \dots, \prio{w}\}$ } (\mmin)_j = |V_{j}| &\\
 & & \text{where $i = \max{\{l \mid (m)_l > 0\}}$ } & \\
 \hline
 w \in V_{\odd} ~&~ \prio{w} \text{ even} ~&~ \mmax =_{\prio{w}} m & \text{(\plh.13)}\\
 \hline
 w \in V_{\odd} ~&~ \prio{w} \text{ odd} ~&~ \mmax =_{i-1} m \,\,\wedge\,\, (\mmax)_{i}=(m)_{i}\!-1 & \text{(\plh.14)}\\
 & & \wedge \text{ for all $j \in \{i+1, \dots, \prio{w}\}$ } (\mmax)_j = |V_{j}| &\\
 & & \text{where $i = \max{\{l \mid (m)_l > 0\}}$ } &\\
 
 \hline
 \end{array}
 \]
 
 \item if $m = \top$, then one of the following holds:
  \[ 
 \begin{array}{|l|l|lr|}
 \hline
 w \in V_{\even} ~&~ \prio{w} \text{ even} ~&~ \mmin = \top & ~\hspace{5pt}\text{(\plh.21)}\\
 \hline 
 w \in V_{\even} ~&~ \prio{w} \text{ odd} ~&~ \mmin = \top \text{or for all $i \leq \prio{w}$ }(\mmin)_i = |V_i| &\text{(\plh.22)}\\
 \hline
 w \in V_{\odd} ~&~ \prio{w} \text{ even} ~&~ \mmax = \top & \text{(\plh.23)}\\
 \hline
 w \in V_{\odd} ~&~ \prio{w} \text{ odd} ~&~ \mmax = \top \text{or for all $i \leq \prio{w}$ } (\mmax)_i = |V_i| & \text{(\plh.24)}\\
 \hline
 \end{array}
 \]
 \end{enumerate}
\end{proposition}

\begin{proof}
 Directly from the definitions of \progname~ and \liftname. Observe that the value $i$ in two subcases of the first part is well-defined, since $\post{(w,m)} \neq \emptyset$.
\end{proof}

Another important property of the Lifting History Graph is that the measure strictly decreases when a LH state with the same vertex is re-encountered along the path in $LH$. 

\begin{restatable}{proposition}{propositionlhgcycledecrease}
 \label{prop:lhg_cycle_decrease}
 If there is a non-trivial path in $LH$ (i.e. containing at least one edge) from $(w,m)$ to $(w,m')$, then $m > m'$.
\end{restatable}

\begin{proof}
 The property is easy to observe, since, intuitively, moving to a successor of $(w,m)$ entails moving to a time spot just before $w$ was lifted to $m$; when $w$ is encountered again, the corresponding snapshot $(w,m')$ comes from some earlier moment, and from the monotonicity of lifting we have $m>m'$ (for a more formal proof, see appendix).
 \qed
\end{proof}

We can use the lifting history graph to define a strategy of player \odd that witnesses some useful capabilities of player \odd: being able to force a certain number of vertices with priority $k$ to appear during the play with no lower (more significant) priority in between, or to force a winning play within a set of priorities bounded by $k$. We call this strategy a lifting history-based (LH-based) strategy $\sigma^{v}_{LH}$. Note that the strategy is not memoryless, and it is of theoretical importance only: its existence serves as a proof of certain properties from which we can in turn derive correctness of our algorithm. 

\def\invertex{v_0}
\def\inkval{kval}
\def\lhstrat{\sigma^{\invertex}_{LH}}

\subsubsection*{LH-based strategy} 

\paragraph{Convention} Throughout the entire section about LH-based strategy, we assume a parity game $G=(V,E,\priosym,(V_{\even},V_{\odd}))$ on which some sequence of liftings has been applied, yielding a temporary, not necessarily stable, measure $\rho$. We also assume a lifting history graph $LH = (V_{LH},E_{LH})$ associated with the aforementioned sequence of liftings performed on $G$.

In such a context, we define the \emph{lifting history-based strategy} $\lhstrat$: a memory-wise partial strategy of player \odd. For ease of presentation, we will present the definition of $\lhstrat$ using an on-line construction procedure (i.e. an algorithm selecting the desired strategy on-the-fly as the play progresses). The procedure utilises a lifting history graph, on which it performs moves in parallel to those in the play. Intuitively we move backwards along the history of updates (liftings) of the corresponding nodes. The measures are thus successively decreased\footnote{Strictly speaking, the measure values do not necessarily decrease with every single \emph{step} in the LH graph, but re-visiting a vertex in LH graph entails a decrease in measure - see Proposition \ref{prop:lhg_cycle_decrease}.}, until a useful (odd) cycle is encountered. We also keep track of the sequence of states in $LH$ visited so far. If the current node in the game is $w \in V_{\odd}$ and the corresponding current state in the lifting history graph is $(w,m)$, the strategy always picks the successor that had the maximal measure value when the current node $w$ was lifted to $m$. Moreover, whenever an odd-dominated cycle is encountered, we remove the entire corresponding suffix from the history and revert to the last state in the lifting history graph that contained $w$.


We proceed with a more formal description of the on-line strategy construction procedure $\oddresponse$, which starts a play at the inital vertex $\invertex$ and, depending on the ownership of the current node, either receives a choice of succesor of player \even, or generates such a choice for player \odd. The procedure maintains the following current state information:
\begin{itemize}
\item $(u,m)$: current state in $LH$, $u$ is the current vertex and $m$ one of its measure values from the lifting history. Initially
$(u,m) = (\invertex,\rho(\invertex))$, where $\rho(\invertex)$ is of the form $(0,|V_1|,0,,\dots,|V_{k-2}|,0,\inkval,***)$, $***$ denoting some arbitrary values.
\item $\lambda = \lambda_1 \dots \lambda_n \in V^{*}$: history of the play (in the parity game $G$) so far, excluding the current vertex, initialised to an empty sequence $\epsilon$
\item $\vis = \vis_1 \dots \vis_{\vislen} \in V_{LH}^{*}$: a sequence of states in $LH$ already visited, initialised to an empty sequence $\epsilon$
\end{itemize}

\def\termcond{\textsc{TermCond}}

$\oddresponse$ may be parametrised with a termination condition $\termcond$, and proceeds as follows:
  
\begin{enumerate}
 \item \label{line:term} If $m=(0,0,\dots,0)$, or $\termcond$ has been met, terminate.
 
 \item \label{line:checkcycle} If an odd-dominated cycle has been encountered, we prune $\vis$ accordingly. That is, if $(u,m') = \vis_j$ for some $j < \vislen$ (we have already visited $u$ and at that point it had a measure $m'$), and moreover the corresponding induced cycle $u = \lambda_i \dots \lambda_{n}. u$ in $G$ is odd-dominated, then we remove the suffix containing the cycle from $\vis$, i.e.  $\vis := \vis_1 \dots \vis_{j-1}.$ Moreover, we replace the current measure value with the previously encountered one, i.e. $(u,m) := (u,m')$.
 
 \item \label{line:update_hist} we update the history: $\lambda := \lambda.u$ and $\vis := \vis.(u,m)$
 
 \item \label{line:evenpicks}If $u \in V_{\even}$, then we receive an input from player \even who picks the next state $x$ from the successors of $u$. We set the current LH state  $(u,m) := (x,m_x)$ such that $(x,m_x) \in \post{(u,m)}$ in $LH$. 

 \item \label{line:oddpicks} Otherwise, if $u\in V_{\odd}$, then we define a choice for player \odd:
 
 $\sigma^{v}_{LH}(\lambda) := x : (x,m_x) \in \post{(u,m)}$ and $m_x$ is maximal within $\post{(u,m)}$ in $LH$. We set $(u,m) := (x,m_x)$.
 
\end{enumerate}

We will now prove several useful properties of the strategy obtained using the $\oddresponse$ procedure. We start with the following lemma concerning the class of strategies that we call oblivious w.r.t. odd-dominated cycles; this class contains the lifting-history based strategy.

\begin{lemma}
 \label{lem:trimming_strats}
 For every vertex $v$ and every (memorywise) strategy $\sigma \in \strategy{\odd}$, oblivious w.r.t. odd-dominated cycles, that is,
 
 \begin{quote}
  $\sigma(\lambda.v.\lambda'.v.\lambda'') = \sigma(\lambda.v.\lambda'')$ whenever $v.\lambda'.v$ is an odd-dominated cycle
 \end{quote}
 
 we have $(\min_{\pi \in \plays{\sigma}{v}}\, \profile{\even}{\pi})\, \in \eventups$.
\end{lemma}

\begin{proof} 
 Suppose $(\min_{\pi \in \plays{\sigma}{v}}\, \profile{\even}{\pi})\, \in \exteventups \setminus \eventups$. Take the play $\pi$ that realises the minimum. We have $(\profile{\even}{\pi})_k > |V_k|$ for some $k$, hence $\pi$ must contain a cycle in the prefix containing a vertex $w$ with $\prio{w}=k$ twice, and not preceded by a priority smaller than $k$. Let $\lambda. w. \lambda'. w$ be such a prefix of $\pi$, and let $w. \lambda'. w$ be a $k$-dominated cycle such that $\prio{w}=k$. Let $\tau$ be a strategy of $\even$ that combined with $\sigma$ yields $\pi$. If we define $\tau'$ as $\tau'(\lambda.w.\lambda'.w.\lambda'') = \tau(\lambda.w.\lambda'')$, and combine it with $\sigma$, we obtain a play $\pi'$ adhering to $\sigma$ such that $(\profile{\even}{\pi'})_k < (\profile{\even}{\pi})_k$, and $(\profile{\even}{\pi'})_i = (\profile{\even}{\pi})_i$ for all relevant $i \neq k$. Hence $\profile{\even}{\pi'} < \profile{\even}{\pi}$, a contradiction with $\pi$ being the play with minimal value among all plays consistent with $\sigma$.
 \qed
\end{proof}

The first key property of the lifting history-based strategy is that it forces every play to have value at least equal to the measure of the initial vertex.

\begin{proposition}
\label{prop:oddresponse_playvalue}
 For every partial play $\pi$ obtained using $\oddresponse$, it holds that $\profile{\even}{\pi} \geq \rho(\invertex)$, or $\pi$ meets $\termcond$.
 
 \end{proposition}

\begin{proof}	
Induction on the number of liftings. The base case $\rho(\invertex) = (0,0,\dots,0)$ is trivial. Suppose that the statement holds for all temporary measures $\rho'$ that occurred before $\rho$. 

Consider a play $\pi$ obtained using $\oddresponse$, and suppose it does not meet $\termcond$. From the way $\oddresponse$ is defined, in case $\rho(\invertex) > (0,0,\dots,0)$, $\pi$ is a sequence of at least two vertices. Let $\pi = \invertex. \pi'$, and $v_1$ be the first vertex of $\pi'$; moreover, let us define 
$m^0 = \rho(\invertex)$, and $m^1 = \rho_{prev}(v_1)$. We need to show that $\profile{\even}{\pi} \geq m^0$.

\mycomment{devise a consistent notation for $\rho$-values, which are in general different for every position/vertex, e.g. $(v_i,\rho_i)$}
 
If $\prio{\invertex}$ is even, then from Proposition \ref{prop:lhg-properties} and the way $\oddresponse$ is defined, we have $m^1 \geq_{\prio{\invertex}} m^0$, which entails $m^1 \geq m^0$ ($(m^0)_i =0$ for $i>\prio{\invertex}$). From IH we know that the value of the suffix $\pi'$ satisfies $\profile{\even}{\pi'} \geq m^1 \geq m^0$. Moreover, since $\profile{\even}{\pi} =_{\prio{\invertex}} \profile{\even}{\pi'}$, and $(m^0)_i =0$ for $i>\prio{\invertex}$, we have $\profile{\even}{\pi} \geq m^0$. 

Suppose now that $\prio{\invertex}$ is odd. We can restrict ourselves to the situation when $m^1 < m^0$, as the case when $m^1 = m^0$ can be proved exactly as above. 

Let $i$ be the largest position such that $(m^0)_i > 0 $. We consider two subcases.
 
 \mycomment{check if comparing $m^0$ and $m^1$ indicates whether carrying occurred - should be ok...}
 
 If $i=\prio{\invertex}$, then $m^0 =_{\prio{\invertex}-1} m^1$, and from $m^1 < m^0$ and Proposition \ref{prop:lhg-properties} we have $(m^0)_{\prio{\invertex}} = (m^1)_{\prio{\invertex}}+1$. On the other hand, we have 
 $\profile{\even}{\pi} =_{\prio{\invertex}-1} = \profile{\even}{\pi'}$, and 
 $(\profile{\even}{\pi})_{\prio{\invertex}} = \profile{\even}{\pi'}_{\prio{\invertex}}+1$ ($\invertex$ contributes one more occurrence of priority $\prio{\invertex}$ to the play value of $\pi$, as compared to $\pi'$). Combining this with IH ($\profile{\even}{\pi'} \geq m^1$), we obtain $\profile{\even}{\pi} \geq m^0$.
 
 If $i<\prio{\invertex}$ (carrying), 
 then it must be the case that $(m^1)_{j} = |V_{j}|$ for all $i<j\leq \prio{\invertex}$.  First let us make an obvious observation that if 
 $\profile{\even}{\pi'} \geq_{\prio{\invertex}-1} m^0$, then $\profile{\even}{\pi} \geq m^0$. 
 \mycomment{since $(m^0)_j=0$ for $j>i$ (?)} 
 
 What thus remains to analyse is the case when $\profile{\even}{\pi'} <_{\prio{\invertex}-1} m^0$. We have then $(\profile{\even}{\pi})_{\prio{\invertex}} > |V_{\prio{\invertex}}|$. At this point, since $\oddresponse$ is oblivious to odd-dominated cycles, we can use Lemma \ref{lem:trimming_strats} -- we know that the minimal play value among all plays obtained with $\oddresponse$ $m_{min} = \min_{\pi'' \in \plays{\sigma_{LH}}{\invertex}}\, \profile{\even}{\pi''} \,\in \eventups$; we also know that 
 $m_{min} \geq m''$, where $m'' =_{\prio{\invertex}-1} m_1$, and $(m'')_{\prio{\invertex}} = |V_{\prio{\invertex}}|+1$. But since in $m_1$ all positions from $i+1$ to $\prio{\invertex}$ are saturated, then for the minimal such value it must hold that $m_{min} >_{i} m_1 (\geq_{i} m_0)$, and hence $\profile{\even}{\pi} \geq m_{\min} > m^0$ \mycomment{reasoning ok, double check details of explanation}.
\qed
\end{proof}

The second key property states that if for some odd priority $k$ we have seen a $k$-dominated stretch of a lesser degree than $(\rho(\invertex))_k$, then the first vertex of priority smaller than $k$ has a measure strictly larger ``up-to-$k$'' than $\rho(\invertex)$.

\begin{proposition}
\label{prop:oddresponse_greater_meausure}
 Every partial play $\pi$ obtained using $\oddresponse$ has the following property: for every relevant odd position $k$, if there is a prefix of $\pi$ of the form $\lambda. v$, where $(\theta(\lambda))_k < (\rho(\invertex))_k$ and $k' = \prio{v} < k$, and moreover if $\lambda. v$ is the smallest prefix with this property, then $\rho(v) >_{k'} \rho(\invertex)$. 
\end{proposition}

\begin{proof}
Induction on the number of liftings. The base case $\rho(\invertex) = (0,0,\dots,0)$ is trivial. Suppose that the statement holds for all temporary measures $\rho'$ that occurred before $\rho$. 

Suppose that there is a prefix of $\pi$ of the form $\lambda. v$, where $(\theta(\lambda))_k < (\rho(\invertex))_k$ and $k' = \prio{v} < k$, and moreover $\lambda. v$ is the smallest prefix with this property. We need to show that $\rho(v) >_{k'} \rho(\invertex)$. 
 
 Let us observe first that $v \neq \invertex$; indeed, from our assumptions it follows that $(\rho(\invertex))_k > 0$, and on the other hand since $\prio{v} < k$, we have $(\rho(v))_k =0$.
 
 Let $v_1$ be the vertex appearing right after $\invertex$ in $\pi$; let us denote with $m^0$ and $m^1$ the values $\rho(\invertex)$ and $\rho(v_1)$, respectively. The suffix of $\lambda$ excluding the first element $v$ will be denoted with $\lambda'$. 
 
 Firstly, observe that if $k' \geq i$ (where $i$ is the least significant position in $\rho(\invertex)$ such that $(\rho(\invertex))_i \neq 0$), then since for all $j>i$ $(\rho(\invertex))_j=0$, and $i \leq k' < k$, we obtain a contradiction with $(\rho(\invertex))_k > 0$. We can therefore restrict ourselves to the case when $k' < i$. 
 
 We consider two cases:
 
 \begin{itemize}
  \item $v = v_1$: From Proposition \ref{prop:lhg-properties} we have $m^1 \geq_{i-1} m^0$; in particular, $m^1 \geq_{k'} m^0$. Suppose, towards contradiction, that $m^1 =_{k'} m^0$. Since $(m^0)_i >0$, and for all $j>k'=\prio{v_1}$ $(m^1)_j = 0$, we have then $m^0 > m^1$. This from Proposition \ref{prop:lhg-properties} excludes the case when $\prio{\invertex}$ is even; in fact, the only possible case is the one in which $i=\prio{\invertex}$ and $(m^0)_i = 1$ with $i$ being the only position in $m^0$ less significant than $k'$ with a non-zero value. But then $i$ is the only possible candidate for $k$, and on the other hand $(\profile{\even}{\invertex.v})_{k} = 1 = \rho(\invertex)$, a contradiction. \mycomment{ok, add a more in-depth explanation in the last part...}
  
  \item $v \neq v_1$: We need to consider several subcases; our goal is to show that in every plausible scenario the IH can be applied to our advantage thanks to some intermediate vertex $v_{I}$ in $\lambda$, occurring between $\invertex$ and $v$, with $\rho(v_I) = m^I$, and choosing some appropriate $k_I$, $k_I > k'$. We can derive the desired property from IH whenever $m^I \geq_{k'} m^0$, and between $v_I$ and $v$ the degree of $k_I$-dominated stretch is smaller than $(m^I)_{k_I}$. Indeed, in such a case we obtain from IH 
  $\rho(v) >_{k'} m^I \geq_{k'} m^0$. 
  
  \paragraph{Observation 1} If there occurs an intermediate node ($w$,$m^w$) between $\invertex$ and $v$ such that $(m^w)_k = |V_k|$, then IH can be applied, and $\rho(v) >_{k'} m^w $.
  
  \paragraph{Proof} From our initial assumptions, we know that $(\profile{\even}{\lambda.v})_k < (m^0)_k$. Since no vertex with priority smaller than $k$ occurs in $\lambda$, for every subplay $\lambda''$ of $\lambda$ we have $(\profile{\even}{\lambda''.v})_k < (m^0)_k$. Consider $\lambda^w$, the suffix of $\lambda$ starting in $w$; we have $(\profile{\even}{\lambda^w.v})_k < (m^0)_k \leq |V_k| = (m^w)_k$. Hence we can apply IH, and obtain $\rho(v) >_{k'} m^w $. \qed
  
  \paragraph{Observation 2} If there occurs an intermediate node ($w$,$m^w$) between $\invertex$ and $v$ such that $(m^w) \geq_{k} m^0$, then $\rho(v) >_{k'} m^0$.
  
  \paragraph{Proof} If $(m^w)_k \geq (m^0)_k$, then IH can be applied immediately \mycomment{ok, explain}. Otherwise we have 
  $m^w >_{k-1} m^0$. If $(m^w)_j > 0$ for any $j$ strictly between $k'$ and $k$, then IH can be applied, since no such priority $j$ can occur until $v$. Otherwise we have $m^w >_{k'} m^0$.
  
  Consider a suffix of $\lambda$ starting at $w$, and let us denote the ensuing vertices with $w = w_0, w_1, \dots, w_N, w_{N+1} = v$. For every $j \in \{0,\dots,N\}$ let us denote with $i_{j}$ the largest (least significant) position such that $\rho(w_j) > 0$. From Proposition \ref{prop:lhg-properties} we know that $\rho(w_j) \geq_{i_j} \rho(w_{j+1})$ for all $j \in \{0,\dots,N\}$. 
  
  Suppose towards contradiction that $\rho(v) \leq_{k'} m^0$. Then $\rho(v) <_{k'} m^w$, and from the above observations there must be some $i_j$ such that $i_j \leq k'$ (otherwise we would have a sequence of inequalities  $\rho(w) = \rho(w_0) \leq_{k'} \rho(w_1) \leq_{k'} \dots \leq_{k'} \rho(w_{N+1}) = \rho(v)$). Since $\prio{w_{i_j}} \geq k$, it must be the case of carrying, and in $\rho(w_{i_j+1})$ all odd positions between $i_j$ and $k$ must be saturated. This excludes $i_j=N$, since $(\rho(w_{N+1}))_l = 0$ for all $l > k'$. Since we have $(\rho(w_{i_j+1}))_k = |V_k|$, we can use IH thanks to Observation 1, and obtain $\rho(v) >_{k'} \rho(w) \geq_{k'} m^0$, from which we derive $\rho(v) >_{k'} m^0$, a contradiction. \qed
  
  We proceed to prove the main statement.
  
  If $k < i$, then $m^1 \geq_{k} m^0$, and the statement follows immediately from Observation 2.  
  
  If $k=i$: if $i=\prio{\invertex}$, then the play value at position $k$ increases by $1$ as compared to $\pi'$; formally $(\profile{\even}{\lambda'.v})_k = (\profile{\even}{\lambda.v})_k-1$. 
  
  If $m^1 \geq_{k} m^0$, then the statement follows from Observation 2. If not, then due to Proposition \ref{prop:lhg-properties}, we have $(m^1)_k = (m^0)_k-1$. We also have 
  $(\profile{\even}{\lambda.v})_k < (m^0)_k$, and from the previous observations we obtain 
  $(\profile{\even}{\lambda'.v})_k < (m^1)_k$, hence IH can be applied with $v_I = v_1, k_I = k$. 
  
  If $k=i$ and $i<\prio{\invertex}$ (carrying), then for all odd $j$ such that $i<j\leq \prio{\invertex}$, we have $(m^1)_{j} = |V_j|$. Consider position $\prio{\invertex}$.

  If $(\profile{\even}{\lambda'.v})_{\prio{\invertex}} < |V_{\prio{\invertex}}|$, then let $w$ be the first vertex such that $\prio{w} < \prio{\invertex}$. From IH we know that $\rho(w) >_{\prio{w}} m^1$. If $w=v$, then 
  $\rho(v) >_{k'} m^1 \geq_{k'} m^0$ (the last inequality holds because $k' < k = i$, and $m^1 \geq_{i-1} m^0$). If $w \neq v$, then, since all positions from $\prio{\invertex}$ up to $k$ are saturated in $m^1$,
  in order to have $\rho(w) >_{\prio{w}} m^1$ it must hold that $\rho(w) >_{k} m^1$, thus
  $\rho(w) \geq_{k} m^0$, and we can apply IH thanks to Observation 2.
  
  If $(\profile{\even}{\lambda'.v})_{\prio{\invertex}} \geq |V_{\prio{\invertex}}|$, then $\lambda$ contains a cycle which gives rise to a $\prio{\invertex}$-dominated stretch of degree at least $|V_{\prio{\invertex}}|+1$. We can then consider a different play that can be obtained with $\oddresponse$, and which does not contain any $\prio{\invertex}$-dominated cycle until $v$ occurs. This gives rise to the previous scenario, and it follows again that $\rho(v) >_{k'} m^0$.
  
 \mycomment{better explanation of the last case needed}
 \end{itemize}
 \qed

\end{proof}
 
If we consider a strategy generated by $\oddresponse$ in which $\termcond$ has been instantiated with 
``a $\top$-labelled vertex has been encountered'', and applied to a vertex whose measure is saturated up-to-$k$, we obtain as a corollary of the above:

\begin{proposition}
\label{prop:lhsbased_strategy_correct}
Assume a parity game $G=(V,E,\priosym,(V_{\even},V_{\odd}))$ on which a sequence of liftings has been applied, resulting in a temporary measure $\rho$. Let $k$ be an \emph{odd} number and let $\invertex$ be a vertex such that for all odd $i<k$, $(\rho(\invertex))_i = |V_i|$, and 
$(\rho(\invertex))_k = \inkval$. 

There exists a strategy $\lhstrat$ of player \odd that guarantees the following objective: for all plays $\pi$ starting at $\invertex$ and conforming to $\lhstrat$, either of the three holds:
 \begin{enumerate}
  \item $\pi$ has a finite prefix that is a $k$-dominated stretch of degree $\inkval$ ($\inkval$ occurrences of vertices with priority $k$)
  \item $\pi$ has a finite prefix that contains only vertices in $V_{\geq k}$, and terminates in a vertex $v$ such that $\rho(v) = \top$
  \item  $\pi$ is infinite, winning for \odd, and contains only vertices in in $V_{\geq k}$ 
 \end{enumerate}
 \end{proposition}
 
\begin{proof}  

Consider $\oddresponse$ with $\termcond$=``a $\top$-labelled vertex has been encountered''. Firstly, observe that the only case when a vertex with priority smaller than $k$ can be encountered, is at the moment when $\oddresponse$ terminates due to an occurrence of a $\top$-labelled vertex. Indeed, as a consequence of the second point of Proposition \ref{prop:oddresponse_greater_meausure}, if such a vertex $w$ occurs, we have $\rho(w) >_{\prio{w}} \rho(\invertex)$, which implies $\rho(w) = \top$ due to $\rho(\invertex)$ being saturated up-to-$k$. 

Suppose that $\termcond$ never held while executing $\oddresponse$. Then due to Proposition \ref{prop:oddresponse_playvalue}, we have $\profile{\even}{\pi} \geq \rho(\invertex)$. If $\profile{\even}{\pi} =_{k-1} \rho(\invertex)$, then $(\profile{\even}{\pi})_{k} \geq (\rho(\invertex))_k = \inkval$, and $\inkval$ vertices of priority $k$ must have been visited. If $\profile{\even}{\pi} >_{k-1} \rho(\invertex)$, then since $\rho(\invertex)$ is saturated up-to-$k$, we have $\profile{\even}{\pi} = \top$, which in turn implies that $\pi$ is an infinite game won by $\odd$ (and staying within $V_{\geq k}$, as shown in the first part of the proof).
\qed
\end{proof}

\begin{corollary}
\label{corr:lhbased_strategy_for_top}
 Assume a parity game $G=(V,E,\priosym,(V_{\even},V_{\odd}))$ on which a sequence of liftings has been applied, resulting in a temporary measure $\rho$ such that $\rho(v) = \top$ for some vertex $v \in V$. There exists a strategy $\sigma$ of player $\odd$ such that for every $\pi \in \plays{v}{\sigma}$ either:
 \begin{enumerate}
  \item $\pi$ visits a vertex $v_{\top}$ such that $\rho(v_{\top}) = \top$; moreover, before visiting $v_{\top}$, only priorities larger or equal to $k$ are encountered
  \item $\pi$ is winning for \odd, and visits only vertices within $V_{ \geq k}$
 \end{enumerate}
 

Moreover, for every successor $u$ of $v$ with a maximal measure among $\post{v}$, there is a strategy $\sigma_u$ with the above properties, $\sigma_u(v)=u$.
\end{corollary}

\begin{proof}

Consider an arbitrary successor $u$ of $v$ with a maximal measure among $\post{v}$; let us define $\sigma$ as  $\sigma_u(v)=u$; for other relevant histories we use the strategy $\sigma^u_{LH}$ from Proposition \ref{prop:lhsbased_strategy_correct}, with the difference that we do not terminate when case 2 has been reached ($(\rho(u))_k$ vertices of priority $k$ visited).

Take any play $\pi \in \plays{v}{\sigma_u}$ and suppose it does not reach any top-labelled vertex (the first case is excluded). According to Proposition \ref{prop:lhsbased_strategy_correct}, there remain two possibilities. Suppose we are in the second case from Proposition \ref{prop:lhsbased_strategy_correct}, i.e.
a $k$-stretch of degree $(\rho(u))_k$ has been visited. Since $\rho(u)$ is saturated up to and including $k$, there were in fact $|V_k|+1$ vertices of priority $k$ visited, including the initial vertex $v$. This means that a certain vertex $w$ has been re-visited by passing through a $k$-dominated cycle. Instead of terminating at this point, we continue the play according to the same lifting history-based strategy, which is oblivious w.r.t. odd-dominated cycles, and stays within $V_{\geq k}$. \mycomment{the last part: state the argument more precisely }

Finally, it is not difficult to observe that if there is a strategy which forces an objective consisting of a disjunction of a winning condition for one of the players, and a reachability objective, then there is a memoryless strategy that guarantees the same objective (one can formally prove this, for instance, using a
straightforward conversion to a winning condition in a parity game).

\qed
 
\end{proof} 
 
\subsection{Existence of the bounded \odd dominion and a partial strategy assignment}
 
\begin{corollary}
\label{cor:v_picks_max_always}
Assume a parity game $G=(V,E,\priosym,(V_{\even},V_{\odd}))$ on which a sequence of liftings has been applied, resulting in a temporary measure $\rho$ such that there is exactly one vertex $v$ with $\rho(v) = \top$. Let $k = \prio{v}$.
  \begin{itemize}
   \item if $v \in V_{\odd}$, then for every successor $u$ of $v$ with a maximal measure among $\post{v}$ there is an \odd-dominion $D_u$ such that for all $w \in D_u$, $\prio{w} \geq k$. Moreover, there is an \odd strategy $\sigma$ winning for \odd, closed on $D_u$, and defined on $v$ as $\sigma(v) = u$
   \item if $v \in V_{\even}$, then there is an \odd-dominion $D$ such that $v \in D$ and for all $w \in D$, $\prio{w} \geq k$. Note that in this case it must hold that $\post{v} \subseteq D$. 
   \end{itemize} 
\end{corollary}
 
 \begin{proof}
 
 Observe that the all assumptions of Corollary \ref{corr:lhbased_strategy_for_top} hold, and from there we know that for every successor $u$ of $v$ with the maximal measure, there exists a memoryless strategy $\sigma_u$ that either yields an infinite play within priorities larger or equal $k$, or will visit a top-labelled vertex. In the latter case, $v$ must be re-visited, as it is the only top-labelled vertex. Hence the desired strategy $\sigma'_u$ for player \odd is the same as  $\sigma_u$, with the difference that if $v$ is re-visited, the history is treated as if it has been reset.
 
 The postulated dominion $D_u$ is the set of all vertices that can occur in plays conforming to $\sigma'_u$. They constitute a dominion of \odd on which no vertex has a priority exceeding $k$.
 \qed
\end{proof}

The above observations are important from an algorithmic perspective, because they allow us to set the strategy of player \odd on the first node $v$ lifted to top while executing the SPM. In fact at this stage we may be able to set an \odd strategy for even more nodes, following a reasoning similar to that in Zielonka's recursive algorithm -- by using a strategy with which \odd can ``attract'' the play from other nodes to $v$. However, to retain soundness, we use a special guarded attractor $\myattr{\geq k}{\odd}{\{v\}}$, which can pass only through nodes of priority not more significant than $k$.

The definition of the guarded attractor given below may be parameterised with a subset of vertices $W \subseteq V$, if we wish 	to consider only part of
the game (in the remainder of the paper, we always use as $W$ a set of vertices inducing a well-defined subgame $G \cap W$). 

If we assume $U \subseteq W \cap V_{\geq k}$, then $\myattrW{\ge k}{\odd}{U}{W}$ is the least set $A$ satisfying:
\begin{enumerate}
 \item $ U \subseteq A \subseteq W \cap V_{\geq k}$
 \item 
 \begin{enumerate}
  \item if $u \in V_{\odd}$ and $\post{u} \cap A \neq \emptyset$, then $u \in A$
  \item if $u \in V_{\even}$ and $\post{u}\cap W \subseteq A$, then $u \in A$
 \end{enumerate}
\end{enumerate}

\newcommand{\stratcomp}[2]{#1 \triangleright #2}

Let $\sigma_1$ and $\sigma_2$ be two strategies of the same player $\player$. By $\stratcomp{\sigma_1}{\sigma_2}$ we will denote a strategy of player $\player$ defined on $\domain{\sigma_1} \cup \domain{\sigma_2}$ as $\sigma_1(w)$ for all $w \in \domain{\sigma_1}$, and as $\sigma_2(w)$ for all $w \in \domain{\sigma_2} \setminus \domain{\sigma_1}$.

\begin{lemma}
\label{lem:extend_strat_with_attr}
Let $D \subseteq G$ be any dominion of \odd in $G$ and $k$ an odd number such that all vertices in $D$ have priority at least $k$, $v \in D$ such that $\prio{v}=k$ and $\sigma_D$ be a winning strategy for \odd on $D$ and closed on $D$. Let $\sigma^{ \geq k}_{Attr}$ be a strategy defined on all vertices in the attractor $\myattr{\geq k}{\odd}{\{v\}} \setminus \{v\}$ as the strategy attracting towards $v$. Then $\stratcomp{\sigma^{ \geq k}_{Attr}}{\sigma_D}$ defined on $D \cup \myattr{\geq k}{\odd}{\{v\}}$ is winning for \odd and only visits priorities greater than or equal to $k$. 
 \end{lemma}
 
\begin{proof}
  Consider an arbitrary infinite play $\pi$, conforming to $\stratcomp{\sigma^{ \geq k}_{Attr}}{\sigma_D}$. If $\pi$ visits $\myattr{\geq k}{\odd}{\{v\}}$ infinitely often, then from the construction of $\stratcomp{\sigma^{ \geq k}_{Attr}}{\sigma_D}$, it will visit $v$ infinitely often, and from the assumption about $D$ the lowest priority in $\pi$ is $k$, hence $\pi$ is winning for $\odd$. Otherwise, $\pi$ has a suffix that stays within $D$, on which it conforms to $\sigma_D$, and therefore is winning for $\odd$ as well. 
\end{proof}
 
Finally, as an immediate consequence of Corollary \ref{cor:v_picks_max_always} and Lemma \ref{lem:extend_strat_with_attr}, we obtain the main result of this section. 
\fi
\noindent
The theorem below forms the basis of our algorithm; it describes
the relevant information about an \odd-dominion that can be extracted
once the first vertex in the game is lifted to top.

\begin{restatable}{theorem}{thmcentraal}
 \label{thm:centraal}
  Let $G$ be a parity game on which an arbitrary lifting sequence is applied, such that at some point a vertex $v$ with $\prio{v}=k$ is the first vertex whose measure value becomes top. Let $\rho$ be the temporary measure at that point. The following holds:
  
  \begin{itemize}
   \item if $v \in V_{\odd}$, then for every successor $u$ of $v$ with a maximal measure among $\post{v}$ there is an \odd-dominion $D_u$ containing $\myattr{\geq k}{\odd}{\{v\}}$ such that for all $w \in D_u$, $\prio{w} \geq k$. Moreover, $\odd$ has a winning strategy that is closed on $D_u$, defined on $v$ as $\sigma(v) = u$, and on $\myattr{\geq k}{\odd}{\{v\}}  \setminus \{v\}$ as the strategy attracting towards $v$
   \item if $v \in V_{\even}$, then there is an \odd-dominion $D$ containing $\myattr{\geq k}{\odd}{\{v\}}$ such that for all $w \in D$, $\prio{w} \geq k$. Moreover,  \odd has a winning strategy $\sigma$ that is closed on $D$, and defined on $\myattr{\geq k}{\odd}{\{v\}}  \setminus \{v\}$ as the strategy attracting towards $v$.  Note that in this case $\post{v} \subseteq D$. 
   \end{itemize}     
     
 \end{restatable} 

\ifreport 
\section{The extended SPM algorithm}
\else
\subsection{The Extended SPM Algorithm}
\fi

\def\res{\textsf{RES}}
\def\rem{\textsf{REM}}
\def\irr{\textsf{IRR}}
\def\dom{\textsf{DOM}}
\def\newspm{\textsf{SPM-Within}}

\begin{algorithm}[h!t]
\caption{Modified SPM with strategy derivation for player Odd}
\label{alg:newspm}
\begin{algorithmic}[1]
\Function{Solve}{$G$}
\State \emph{\textbf{Input} $G = (V, E, \priosym, (V_\even, V_\odd))$}
\State \emph{\textbf{Output} Winning partition and strategies $((\winsubeven{G},\sigma'),(\winsubodd{G},\sigma))$}
\State initialise $\sigma$ to an empty assignment
\State $\rho  \gets \lambda w \in V.~(0, \dots, 0)$
\State $\Call{\newspm}{V}$
\State compute strategy $\sigma'$ for player Even by picking min. successor w.r.t. $\rho$
\State \Return $((\{v \in V ~|~ \rho(v) \neq \top\},\sigma'),(\{v \in V ~|~ \rho(v) = \top,\sigma))$

\State

\Procedure{\newspm}{$W$}  

\While{$(W \neq \emptyset)$} \label{line:outerwhile}
\While{$\rho \sqsubset \lift{\rho}{w} \text{ for some $w \in W$ and for all $w \in W$:} \rho(w) \neq \top$} \label{line:liftloop_start}
\State $\rho \gets \lift{\rho}{w}$ for $w \in W$ such that $\rho \sqsubset \lift{\rho}{w}$
\EndWhile \label{line:liftloop_end}
\State \textbf{if} {$\text{for all $w \in W$: } \rho(w) \neq \top$} \textbf{return} $\rho$ \textbf{end if}

\State \label{line:firsttop_v} $v \gets $ the (unique) vertex in $W$ such that $\rho(v) = \top$
\State $k \gets \prio{v}$
\State \label{line:firsttop_strat} $\sigma(v) \gets u$ where $u \in v^\bullet \cap W$ for which $\rho(u') \leq_{k} \rho(u)$ for all $u' \in v^\bullet \cap W$
\State $\res \gets \myattrW{\geq k}{\odd}{\{v\}}{W}$
\For{ \textbf{all} $w \in \res \setminus \{v\}$}
\State $\rho(w) \gets \top$
\State \textbf{if} $w \in V_{\odd}$ \textbf{then} assign $\sigma(w)$ the strategy \emph{attracting to $v$} \textbf{end if} 
\EndFor \label{line:defsigma_res_end}
\State $\dom \gets \res$
  \State \label{line:compute_irr} $\irr \gets \attrW{\even}{\{ w \in W \,\mid\, \prio{w} < k\}}{W}$
  \State \label{line:rem} $\rem \gets W \setminus (\res \cup \irr)$
  \State \label{line:reccall} $\Call{\newspm}{\rem}$  
  \State $\dom \gets \dom \cup \{ w \in \rem \,\mid\, \rho(w) = \top \}$ \label{line:dominion_resolved}
\State $A \gets \attrW{\odd}{\dom}{W}$

  \For{ \textbf{all} $w \in A \setminus \dom$}
  \State $\rho(w) \gets \top$
    \If{$w \in V_{\odd}$} 
    assign $\sigma(w)$ to be the strategy  \emph{attracting} to $\dom$ 
    \EndIf
  \EndFor
\State \label{line:smallerW} $W \gets W \setminus A$
\EndWhile
\EndProcedure
\EndFunction

\end{algorithmic}
\end{algorithm}

\newcommand{\genlift}[3]{\ensuremath{\text{lift}_{#1}(#2,#3)}}
\def\genliftname{\text{lift}_{W}}

Theorem \ref{thm:centraal} captures the core idea of our algorithm.
It provides us with the means to locally resolve (\ie define a local
strategy for) at least one vertex in $\winsubodd{G}$, once a top value
is found while lifting. Moreover, it indicates in which part of
the game the remainder of the \odd-dominion resides, implying
that one can temporarily restrict the lifting to that area until
the dominion is fully resolved. \ifreport There is still a non-trivial task
ahead: how to proceed such that the composition of all local strategy
assignments will be globally valid.
We will give a description of our solution (Algorithm \ref{alg:newspm}), and
informally argue the correctness of our approach. 
The formal correctness proof can be found in section \ref{sec:correctness}.
\else
We will give a description of our solution (Algorithm \ref{alg:newspm}), and
informally argue the correctness of our approach.  For a (intricate and rather involved) formal proof, we
refer to~\cite{GW:14}.
\fi

The algorithm proceeds as follows. First, a standard SPM is run until the first vertex reaches top [l. \ref{line:liftloop_start}--\ref{line:liftloop_end} in Alg. \ref{alg:newspm}\,]. Whenever $v$ is the first vertex lifted to top, then the issue of the winning strategy for $v$ can be resolved immediately [l. \ref{line:firsttop_strat}\,], as well as for vertices in the `at-least-$k$' attractor of $v$ (if there are any). We will denote this set of `resolved' vertices with $\res$. Moreover, we can restrict our search for the remainder of the \odd-dominion $D$ only to vertices with priorities not more significant than $k$, in fact only those from which player \even cannot attract a play to visit a priority more significant than $k$. Hence we can remove from the current computation context the set $\irr = \myattr{}{\even}{\{ w \in W \,\mid\, \prio{w} < k\}}$, vertices that may be considered at the moment irrelevant [l. \ref{line:compute_irr}--\ref{line:rem}\,].

After discarding the resolved and currently irrelevant vertices, the algorithm proceeds in the remaining set of vertices that constitutes a proper subgame (i.e. without dead ends) induced by the set $\rem$. After the subroutine returns [l. \ref{line:reccall}\,], all vertices labelled with top are won by player \odd in the subgame $G \cap \rem$. In other words, those vertices are won by \odd provided that the play does not leave $\rem$. Since the only way for player \even to escape from $\rem$ is to visit $\res$, where every vertex is won by player \odd, the top-labelled vertices from $\rem$ are in fact won by \odd in the context of the larger game $G \cap W$. Therefore the set $\dom$ computed in line \ref{line:dominion_resolved} constitutes an \odd-dominion within the game $G \cap W$, in which we have moreover fully defined a winning strategy $\sigma$ for player \odd. Finally, every vertex from $V \setminus \dom$ that can be attracted by player \odd to the dominion $\dom$ is certainly won by \odd, and for those vertices we assign the standard strategy attracting to $\dom$. The \odd-dominion $A$ is then removed, and the computation continues in the remaining subgame. 

The algorithm may at first sight appear to deviate much from the
standard SPM algorithm. However, the additional overlay, apart from
defining the strategy, can be seen as a special lifting policy that
temporarily restricts the lifting to parts where an \odd dominion
resides. 

\ifreport

\subsection{Correctness of the modified algorithm}
\label{sec:correctness}

The core of the correctness proof consists of showing that upon return procedure $\newspm$ computes the winning strategy for player \odd for all vertices in $\winsubodd{G \cap W}$, provided that the input set $W$ meets certain guards, which we call \emph{suitability} conditions. 

We need to define a few notions first. Let $A \subseteq V$ in the lifting context $\tuple{G,\rho}$. We will say that $A$ \emph{has only nonprofitable \odd-escapes with respect to $\rho$} if for every $w \to u$ such that $w \in A \cap V_{\odd}$ and $u \in V \setminus A$, it holds that   $u \in \myattrW{}{\even}{\{ w' \in V \,\mid\, \prio{w'} < \minprio{A}\}}{\nontop{\rho}} $, where $\nontop{\rho} = \{w' \in V \mid \rho(w') \neq \top\}$.

Moreover, we will say that $A$ \emph{has only top \even-escapes with respect to $\rho$}, if for every $w \to u$ such that $w \in A \cap V_{\even}$ and $u \in V \setminus A$, it holds that $\rho(u)=\top$.

We will call a subset $W \subseteq V$ \emph{suitable} w.r.t $\rho$ if:

\begin{itemize}
   \item[S1] for all $w \in W$, $\rho(w) \neq \top$
   \item[S2] $G \cap W$ is a proper sugbame
   \item[S3] $W$ has only nonprofitable \odd-escapes w.r.t. $\rho$
   \item[S4] $W$ has only top \even-escapes w.r.t. $\rho$
\end{itemize}

Theorem \ref{thm:centraal} captures the main idea behind our algorithm, and for the sake of understanding and readability it was singled out in a simplified form, as compared to the version that is formally required to prove correctness of the algorithm. The latter version is the following generalisation of Theorem \ref{thm:centraal}, that allows us to reason about a context in which possibly more than one top value has occurred in the course of the lifting and  certain parts of the game have already been resolved.

\def\oldrho{\bar{\rho}}
\def\rhomaxout{\rho_{max-out}}

\begin{restatable}{theorem}{thmcentraalgeneralised}
 \label{thm:centraal-generalised}
 Suppose that $W \subseteq V$ induces a proper subgame and $W$ has only nonprofitable \odd-escapes. Let $\rho$ be a measure corresponding to a lifting sequence in which $v \in W$ is the only vertex in $W$ such that $\rho(v) =\top$, and $v$ was the last lifted vertex. Let $k = \prio{v}$. Then there is a (memoryless) strategy $\sigma$, winning for \odd in the context of the subgame $G \cap W$, such that all plays conforming to $\sigma$ visit only vertices with priorities not smaller than $k$. Moreover, if $v \in V_{\odd}$, then $\sigma(v)$ is (one of the) maximal successor(s) of $v$ w.r.t. $\rho$.
 
 \myomit{
 Moreover, on $\myattrW{\geq k}{\odd}{\{v\}}{W} \cap V_{\odd}$, $\sigma$ is defined as:
 \begin{itemize}
  \item for all $w \in \myattrW{\geq k}{\odd}{\{v\}}{W} \setminus \{v\}$, $\sigma$ is the corresponding strategy attracting to $v$
  \item if $v \in V_{\odd}$, then $\sigma(v)$ is (one of the) maximal successor(s) of $v$ w.r.t. $\rho$
 \end{itemize}
}
\end{restatable}

The key property concerning correctness of Algorithm \ref{alg:newspm} is proved in Proposition \ref{prop:algcorrectness}. The proposition utilises several lemmata, which we state below. Their proofs can be found in the appendix.

\begin{restatable}{lemma}{lemmaargissubgame}
\label{lem:arg_is_subgame}
 If $W$ is a subgame, then the set $\rem = W \setminus (\res \cup \irr)$ that is computed in line \ref{line:rem} induces a subgame (i.e. it does not have ``dead ends'').
\end{restatable}

\begin{restatable}{lemma}{lemmarestrlifttop}
\label{lem:restr_lift_top}
 Suppose some arbitrary lifting procedure has been applied on the entire $G$, yielding a temporary measure $\oldrho$  Assume that $W$ is a set of vertices that induces a well-defined subgame of $G$, $G \cap W$, and moreover the only edges leading from the even-owned vertices in $W$ to $G \setminus W$, have top-labelled vertices as endpoints. Furthermore, suppose that $D \subseteq W$ induces an \odd-dominion on the subgame $G \cap W$. Then lifting of $\rho$ restricted to $W$ will finally yield a top value.
\end{restatable}

\begin{restatable}{lemma}{lemmasubdominion}
\label{lem:subdominion}
 Let $D$ be an $\odd$-dominion within a game $G$. Let $D' \subseteq D$ be a nonempty subset of $D$ such that $D'$ has only \even-escapes to $D \setminus D'$. That is, for all $u \to w$ such that 
 $u \in D'$ and $w \in D \setminus D'$, it holds that $u \in V_{\even}$. Then $D'$ is a dominion within any subgame $G'$ containing the entire $D'$, but not containing any vertices from $D \setminus D'$.
\end{restatable}

We are now in the position to prove the key result of this section (here, we provide a high-level description of the main steps of the proof, and for the details we refer to the appendix).

\begin{restatable}{proposition}{propalgcorrectness}
 \label{prop:algcorrectness}
  
  Assume a lifting context $\tuple{G,ms}$. Suppose that during the execution of the procedure $\newspm$, before some while-loop iteration (line \ref{line:outerwhile}), $\rho$ has a certain value $\rho_I$, and it holds that $W$ is suitable w.r.t. $\rho_I$. Let $\rho_F$ be the final value of $\rho$ when $\newspm$ returns. Then after executing $\newspm$, the following hold:
  \begin{itemize}
   \item for all $w \in W$, $w \in \winsubodd{G \cap W} \,\iff\, \rho_F(w) = \top$
   \item $\sigma|_{\winsubodd{W}}$ is winning for \odd in the context of a subgame $G \cap W$
  \end{itemize}
 \end{restatable}
 
 \begin{proof}   
 
 We proceed with structural induction on $W$; assume that the statement holds for all suitable subsets of $W$. We will prove that it holds for $W$.
 
 \begin{itemize}
   \item[I] If $\winsubodd{G \cap W} \neq \emptyset$, then the iteration of liftings in lines \ref{line:liftloop_start}--\ref{line:liftloop_end} will eventually lead to a top-value in some vertex $v$.    
   
   \item[II] Vertex $v$ defined in line \ref{line:firsttop_v} belongs to $\winsubodd{G \cap W} $

   \item[III] In line \ref{line:rem}, $\rem$ is suitable w.r.t. $\rho$
   
   \item[IV] For any $D \subseteq (\res \cup \rem)$ which is an \odd dominion in the context of $G \cap W$, all vertices in $D \setminus \res$ are also won by \odd in $G \cap \rem$, i.e. 
   $D \setminus \res \subseteq \winsubodd{G \cap \rem}$

   \item[V] $\sigma|_{\rem}$ is winning for \odd on $\winsubodd{G \cap \rem}$

   \item[VI] $\sigma|_{\res \cup \winsubodd{G \cap \rem}}$ is a winning strategy for $\odd$ in $\winsubodd{G \cap \rem} \cup \res$ in the context of the subgame $G \cap W$

   \item[VII] $\sigma|_{A}$ is a winning strategy for $\odd$ in $A$ in the context of the subgame $G \cap W$

   \item[VIII] if the new $W^{new} := W \setminus A$, computed in line \ref{line:smallerW}, is not empty, then it is suitable w.r.t. $\rho$. 

  \end{itemize}

 \end{proof}

\fi

\ifreport

\begin{restatable}{theorem}{thmcorrectness}
\label{thm:correctness}
$\newspm$ returns the least game progress measure of $G$, and
the strategy $\sigma$, fully defined on $\winsubodd{G}
\cap V_{\odd}$, is a winning strategy of player \odd on $\winsubodd{G}$.

\end{restatable}

\begin{proof}
 As $V$ is obviously suitable w.r.t. the initial $\rho = \lambda v.(0,\dots,0)$, from Proposition \ref{prop:algcorrectness} we immediately obtain correct resolution of the $\winsubodd{G}$ part.
\end{proof}

\paragraph*{Running time} 
\fi 
Every attractor computation takes $O(n+m)$
time, and whenever it occurs, at least one new vertex is `resolved'.
Hence the total extra time introduced by the attractor computations
is bounded by $O(n(n+m))$. As with the standard SPM, the lifting
operations dominate the running time, and their total number for
every vertex is bounded by the size of $\mathbb{M}^{\even}$.
\begin{theorem}
For a game $G$ with $n$ vertices, $m$ edges, and
$d$ priorities, $\textsc{Solve}$ solves $G$ and computes winning
strategies for player $\even$ and $\odd$ in worst-case $\runtimefloor$.
\end{theorem}
\ifreport
\begin{proof}
\begin{equation*}
O\left(n\left(n+m \right) + dm \left(\frac{n}{\lfloor \frac{d}{2} \rfloor}\right)^{\lfloor \frac{d}{2} \rfloor} \right) = O\left(dm \left(\frac{n}{\lfloor \frac{d}{2} \rfloor}\right)^{\lfloor \frac{d}{2} \rfloor} \right)
\end{equation*}
\end{proof}
\fi

\begin{example}
We illustrate the various aspects of Algorithm~\ref{alg:newspm} on
the game $G$ depicted in Figure~\ref{fig:illustrating_example}, with two (overlapping) subgames
$G_1$ and $G_2$.
Note that the entire game is an $\odd$-paradise: every vertex
eventually is assigned measure $\top$ by Algorithm~\ref{alg:newspm}
(and Algorithm~\ref{alg:spm}, for that matter). Suppose we use a
lifting strategy prioritising $v_2,v_3,v_7$ and $v_8$; then vertex
$v_3$'s measure is the first to reach $\top$, and the successor
with maximal measure is $v_7$. Therefore, $\odd$'s strategy is to
move from $v_3$ to $v_7$.  The set $\res$, computed next consists of
vertices $v_3$ and $v_2$; the strategy for $v_2$ is set to $v_3$
and its measure is set to $\top$. The $\even$-attractor into those
vertices with priorities $\ge 3$, \ie, vertices
$v_1$ and $v_4$, is exactly those vertices, so, next, the algorithm
zooms in on solving the subgame $G_1$.
\begin{figure}[h!tbp]
\centering
\begin{tikzpicture}[>=stealth']
\tikzstyle{every node}=[draw, inner sep=0pt, outer sep=0pt, node distance=40pt];
  \node[draw=none,shape=diamond,minimum size=26]   (s1)                                 {};
  \node[draw=none,shape=rectangle,minimum size=22] (s2) [right of=s1,xshift=20pt]                   {};
  \node[draw=none,shape=rectangle,minimum size=22] (s3) [right of=s2,xshift=20pt]                   {};
  \node[draw=none,shape=diamond,minimum size=26]   (s4) [right of=s3,xshift=80pt]                   {};
  \node[draw=none,shape=rectangle,minimum size=22] (s5) [below of=s4]                   {};
  \node[draw=none,shape=diamond,minimum size=26]   (s6) [left of=s5,xshift=-20pt]                   {\scriptsize 5};
  \node[draw=none,shape=rectangle,minimum size=22] (s7) [left of=s6,xshift=-20pt]                   {};
  \node[draw=none,shape=diamond,minimum size=26]   (s8) [left of=s7,xshift=-20pt]                   {};
  \node[draw=none,shape=rectangle,minimum size=22] (s9) [left of=s8,xshift=-20pt]                   {};


\draw [fill=gray!15, thick,rounded corners=55pt, fill opacity=0.5] 
    ($(s5)+(1.7,-0.8)$) -- 
    ($(s4)+(1.0,1.65)$) -- 
    ($(s6)+(-2.2,-0.8)$) -- cycle
;
\draw 
    [fill=gray!10, thick,rounded corners=10pt, fill opacity=0.3] 
    ($(s5)+(0.7,-1.1)$) -- 
    ($(s9)+(-1.2,-1.1)$) -- 
    ($(s9)+(-1.2,0.6)$) -- 
    ($(s5)+(0.7,0.6)$) -- cycle;

\draw ($(s4)+(1.0,-0.5)$) node[draw=none] {$G_2$};
\draw ($(s9)+(-0.8,-0.7)$) node[draw=none] {$G_1$};

  \node[label=above:{$v_1$},shape=diamond,minimum size=26]   (y1)                                 {\scriptsize 0};
  \node[label=above:{$v_2$},shape=rectangle,minimum size=22] (y2) [right of=y1,xshift=20pt]                   {\scriptsize 4};
  \node[label=above:{$v_3$},shape=rectangle,minimum size=22] (y3) [right of=y2,xshift=20pt]                   {\scriptsize 3};
  \node[label=above:{$v_4$},shape=diamond,minimum size=26]   (y4) [right of=y3,xshift=80pt]                   {\scriptsize 1};
  \node[label=below:{$v_5$},shape=rectangle,minimum size=22] (y5) [below of=y4]                   {\scriptsize 4};
  \node[label=below:{$v_6$},shape=diamond,minimum size=26]   (y6) [left of=y5,xshift=-20pt]                   {\scriptsize 5};
  \node[label=below:{$v_7$},shape=rectangle,minimum size=22] (y7) [left of=y6,xshift=-20pt]                   {\scriptsize 5};
  \node[label=below:{$v_8$},shape=diamond,minimum size=26]   (y8) [left of=y7,xshift=-20pt]                   {\scriptsize 6};
  \node[label=below:{$v_9$},shape=rectangle,minimum size=22] (y9) [left of=y8,xshift=-20pt]                   {\scriptsize 4};

\path[->]
  (y1) edge[bend left] (y3)
  (y2) edge (y3) edge[bend left] (y8)
  (y3) edge (y7)
  (y4) edge (y3) edge[loop right] (y4)
  (y5) edge (y4) edge[bend left] (y6)
  (y6) edge[bend left] (y5) edge[bend left] (y7)
  (y7) edge [bend left] (y6) edge[bend left] (y8) edge[bend left] (y9)
  (y8) edge [bend left] (y2) edge[bend left] (y7)
  (y9) edge (y1) edge[loop left] (y9)
;

\end{tikzpicture}
\caption{An example game $G$ with two (overlapping) subgames $G_1$ and $G_2$.}
\label{fig:illustrating_example}
\end{figure}

Suppose that within the latter subgame, we prioritise the lifting of
vertex $v_7$ and $v_8$; then vertex $v_7$'s measure is set to $\top$
first, and $v_7$'s successor with the largest measure is $v_8$;
therefore $\odd$'s strategy is to move from $v_7$ to $v_8$. At this
point in the algorithm, $\res$ is assigned the set of vertices $v_7$
and $v_8$, and the measure of $v_8$ is set to $\top$. Note that in this
case, in this subgame, the winning strategy for $\odd$ on $v_7$ is to remain within
the set $\res$. Since all
remaining vertices have more signficant priorities than $v_7$, or
are forced by $\even$ to move there, the next recursion is run on
an empty subgame and immediately returns without changing the
measures. Upon return, the $\odd$-attractor to all $\odd$-won
vertices (within the subgame $G_1$, so these are only the 
vertices $v_7$ and $v_8$)
is computed, and the algorithm continues solving the remaining
subgame (\ie the game restricted to vertices $v_5,v_6$ and $v_9$),
concluding that no vertex in this entire game will be assigned
measure $\top$.

At this point, the algorithm returns to the global game again and
computes the $\odd$-attractor to the vertices won by player $\odd$ at that stage
(\ie vertices $v_2,v_3,v_7$ and $v_8$), adding vertices $v_1$ and $v_9$,
setting their measure to $\top$ 
and setting $\odd$'s strategy for $v_9$ to move to $v_1$. 

As a final step, the algorithm next homes in on the subgame $G_2$,
again within the larger game.  The only vertex assigned measure
$\top$ in the above subgame is vertex $v_4$; at this point $\res$ is
assigned all vertices in $G_2$, the measure of $v_5$ and $v_6$ is set to
$\top$ and the $\odd$ strategy for vertex $v_5$ is set to $v_4$.
This effectively solves the entire game.

\end{example}

\section{Conclusions and Future Work}
The two key contributions of our work are:
\begin{enumerate}
 \item We have proposed a more operational interpretation of progress measures by characterising the types of plays that players can enforce. 
 \item We have provided a modification of the SPM algorithm that allows to compute the winning strategies for both players in one pass, thus improving the worst-case running time of strategy derivation.
 \end{enumerate}
The second enhancement has been made possible due to a thorough study of the contents of progress measures, and their underapproximations in the intermediate stages of the algorithm (building on the proposed operational interpretation).

As for the future work, we would like to perform an analysis of SPM behaviour on special classes of games, along the same lines as we have done in case of the recursive algorithm \cite{GW:13}, specifically, identifying the games for which SPM runs in polynomial time, and studying enhancements that allow to solve more types of games efficiently.

\bibliographystyle{plain}
\bibliography{spm}

\newpage

\appendix


\ifreport

\section*{Appendix}

\subsubsection{Proofs of Section \ref{sec:lhg}}

\propositionlhgcycledecrease*

\begin{proof}
 Suppose there is a path in $LH$ $(w,m) = (w_0,m_0) \to (w_1,m_1) \to \dots \to (w_n,m_n) =(w,m')$. Let $\rho_{t}$ denote the value of $\rho$ after the $t$-th lifting; for $j \in \{0, \dots , n\}$ we will define $t_j$ as the step in which $w_j$ was lifted to $m_j$. From monotonicity of lifting and the fact that two vertices cannot be lifted at the same moment, it is not difficult to observe that $t_{j} > t_{j+1}$, and hence from transitivity we have $t_0 > t_n$
 We consider two cases:
\begin{itemize}
 \item if $n>1$, then \[ m = \rho_{t_0}(w_0) \stackrel{\text{def. of $t_0$}}{>} \rho_{t_0-1}(w_0) \stackrel{t_0-1 \geq t_n}{\geq} \rho_{t_n}(w_0) = \rho_{t_n}(w_n) = m'\]
 
 \item if $n=1$, then \[m = \rho_{t_0}(w_0)  \stackrel{\text{def. of $t_0$}}{>} \rho_{t_0-1}(w_0) = \rho_{t_n}(w_n) = m'\]
 
\end{itemize}
\end{proof}

\thmcentraalgeneralised*
  
\begin{proof}
 Consider the strategy $\sigma$ from Corollary \ref{corr:lhbased_strategy_for_top}. No play conforming to $\sigma$ can visit a vertex with priority smaller than $k$, unless it is a top-labelled node occurring after a prefix of non-top nodes having priorities at least $k$. Therefore no such play can enter 
 $\myattrW{}{\even}{\{ w' \in V \,\mid\, \prio{w'} < k \}}{\nontop{\rho}}$, which is a subset of $\myattrW{}{\even}{\{ w' \in V \,\mid\, \prio{w'} < \minprio{W} \}}{\nontop{\rho}}$. From this and the fact that $W$ has only non-profitable \odd-escapes, we know that $\sigma(w) \in W$ for all $w \in \domain{\sigma} \cap W$. Hence $\sigma$ is a well-defined strategy in the context of the subgame $G \cap W$.
 
 Consider any play $\pi$ conforming to $\sigma$, and staying within $G \cap W$. In case when $\pi$ is finite and visits the first top-labelled vertex, it can only be $v$, on which $\sigma$ is defined.
 If $\pi$ visits $v$ inifinitely often, then $\pi$ is winning for \odd, because $\pi$ does not visit any more significant priority in between. If $v$ is visited only finitely many times, from Corollary \ref{corr:lhbased_strategy_for_top} we know that $\pi$ is winning for \odd. In both cases no priority more significant than $k$ is encountered.

\end{proof}

\subsubsection{Proofs of Section \ref{sec:correctness}} 
 
 \lemmaargissubgame*
 
 \begin{proof}

 Since $G$ is a well-defined game, we know that $\post{u} \neq \emptyset$. Suppose, towards contradiction, that there is a node $u \in V \setminus (\res \cup \irr)$ such that $\post{u} \subseteq \res \cup \irr$. We distinguish two cases:

\begin{itemize}
 \item if $u \in V_{\even}$, then, since $V \setminus \irr$ is an \odd-closed set, we have $\post{u} \cap \irr = \emptyset$, so the only possibility is that $\post{u} \subseteq \res$. However, since $u \in V \setminus (\res \cup \irr)$, we have $\prio{u} \geq k$ and hence it must be the case that $u \in \myattr{\geq k}{\odd}{\res} = \res$, a contradiction.
 
 \item if $u \in V_{\odd}$, then $\post{u} \cap \res = \emptyset$ (because $\prio{u} \geq k$ and in that case we would have $u \in \myattr{\geq k}{\odd}{\res} = \res$). Therefore $\post{u} \subseteq \irr$ - but it means that $u \in  \myattr{}{\even}{\irr} = \irr$, a contradiction.
\end{itemize}
\end{proof}

\lemmarestrlifttop*

\begin{proof}
 We can transform $G$ to a game $G'$, in which we remove all edges from $W$, to $V \setminus W$. The codomain of SPM for $G'$ is the same as for $G$. Lifiting in $G'$ must reach a top because of $D$, and it will yield smaller values that lifting in $G$. Indeed, the removed edges leading from $W \cap V_{\even}$ (vertices on which min is taken) outside lead only to top-labelled vertices, so the effect is the same as removing these edges. Edges originating in $W \cap V_{\odd}$ may only increase the measure.
\end{proof}

The purpose of the following lemma is to establish that the remainder of the dominion containing $\res$, and contained in $\rem$, is a dominion within $\rem$. Note that the subgame $G'$ mentioned in the lemma may actually not exist (however, in our case it always exists ($\rem$)-- see Lemma \ref{lem:arg_is_subgame}).

\lemmasubdominion*

\begin{proof} 
 Consider the \odd strategy $\sigma$, winning for \odd, and closed on $D$. Since there are no \odd-escapes from $D'$ to $D \setminus D'$, for any $w \in \domain{\sigma} \cap D'$ we have $\sigma(w) \in D'$. Hence for any subgame $G'$ containing $D'$, but not any vertex from $D \setminus D'$, the strategy $\sigma$ restricted to $D'$ is well-defined, and the two desired properties are carried over from the original strategy.
\end{proof}

\propalgcorrectness*

 \begin{proof}   
 
 We proceed with structural induction on $W$; assume that the statement holds for all suitable subsets of $W$. We will prove that it holds for $W$.
 
 \begin{itemize}
   \item[I] If $\winsubodd{G \cap W} \neq 0$, then the iteration of liftings in lines \ref{line:liftloop_start}--\ref{line:liftloop_end} will eventually lead to a top-value in some vertex $v$. 
   
   \emph{Proof} Follows from Lemma \ref{lem:restr_lift_top} and the assumption of $W$ having only top \even-escapes.
   
   \item[II] Vertex $v$ defined in line \ref{line:firsttop_v} belongs to $\winsubodd{G \cap W} $
   
   \emph{Proof} Immediately from Theorem \ref{thm:centraal-generalised}.
   \item[III] In line \ref{line:rem}, $\rem$ is suitable w.r.t. $\rho$
   
   \emph{Proof}
   \begin{enumerate}
    \item[S1] obvious, since $v$ is the only vertex in $W$ such that $\rho(v)=\top$, and $v \notin \rem$
    \item[S2] proved in Lemma \ref{lem:arg_is_subgame}
    \item[S3] from the assumption about $W$, all \odd escapes from $\rem$ outside $W$ are non-profitable. Consider any \odd-escape to $W \setminus \rem$, that is, $u \to w$ such that $u \in \rem \cap V_{\odd}$, and 
    $u \in W \setminus \rem = \irr \cup \res$. Suppose towards contradiction that $w \notin \irr$, hence $w \in \res$. But since $\prio{u} \geq k$ and $u \in W \cap V_{\odd}$, we have $u \in \res$, and therefore $u \notin \rem$, a contradiction.
    \item[S4] from the assumption about $W$, all \even-escapes from $\rem$ outside $W$ are to top-labelled nodes. Observe that no even-owned vertex $w \in \rem$ can have an edge to $\irr = \attrW{\even}{\{ w \in W \,\mid\, \prio{w} < k\}}{W}$, because $w$ would then belong to $\irr$. 
    Hence the only \even-escapes outside $\rem$ lead to $\res$, and every vertex therein has measure top.    
   \end{enumerate}
   
   \item[IV] For any $D \subseteq (\res \cup \rem)$ which is an \odd dominion in the context of $G \cap W$, all vertices in $D \setminus \res$ are also won by \odd in $G \cap \rem$, i.e. 
   $D \setminus \res \subseteq \winsubodd{G \cap \rem}$
   
   \emph{Proof} Follows from Lemma \ref{lem:subdominion} and the fact that there are no \odd-escapes from $\rem$ to $\res$.
   
   \item[V] $\sigma|_{\rem}$ is winning for \odd on $\winsubodd{G \cap \rem}$
   
   \emph{Proof} Follows from the inductive hypothesis.

   \item[VI] $\sigma|_{\res \cup \winsubodd{G \cap \rem}}$ is a winning strategy for $\odd$ in $\winsubodd{\rem} \cup \res$ in the context of the subgame $G \cap W$
   
   \emph{Proof} Follows from Lemma \ref{lem:extend_strat_with_attr} and the previous point.
   
   \item[VII] $\sigma|_{A}$ is a winning strategy for $\odd$ in $A$ in the context of the subgame $G \cap W$
   
   \emph{Proof} Follows from the previous point and the obvious fact that extending the dominion with its attractor, and assigning the attracting strategy for the extended part yields a winning strategy.

   \item[VIII] if the new $W := W \setminus A$, computed in line \ref{line:smallerW}, is not empty, then it is suitable w.r.t. $\rho$
   
   \begin{enumerate}
    \item[S1] by removing $A$, all top-labelled vertices have been removed from $W$
    \item[S2] $W$, as a complement of an \odd-attractor, is \even-closed, and constitutes a proper subgame
    \item[S3] the new $W$ doesn't have any additional \odd-escapes as compared to the old one
    \item[S4] the only possible new \even-escapes lead to $A$, in which all vertices are top-labelled
   \end{enumerate}

  \end{itemize}
 \end{proof}

\fi
  
\end{document}